\documentclass[sigconf, nonacm]{acmart}

\newcommand\vldbdoi{XX.XX/XXX.XX}
\newcommand\vldbpages{XXX-XXX}
\newcommand\vldbvolume{14}
\newcommand\vldbissue{1}
\newcommand\vldbyear{2020}
\newcommand\vldbauthors{\authors}
\newcommand\vldbtitle{\shorttitle} 
\newcommand\vldbavailabilityurl{URL_TO_YOUR_ARTIFACTS}
\newcommand\vldbpagestyle{plain} 

\usepackage{graphicx}
\usepackage[utf8]{inputenc} 
\usepackage{amsmath,amsthm} 
\usepackage{url}
\usepackage{fancyvrb}
\usepackage{wrapfig}
\usepackage{xspace}
\usepackage{subcaption}
\usepackage{booktabs}
\usepackage{hyperref}
\usepackage{tikz-cd}
\usepackage{tikz}
\usepackage{neuralnetwork}
\hypersetup{
colorlinks=false,
citebordercolor={green}
}

\newcommand{\direct}{{\sc histogram-based}\xspace}
\newcommand{\randomwalk}{{\sc random-walk}\xspace}
\newcommand{\onlineunion}{{\sc online-union}\xspace}

\newcommand{\todo}[1]{\textcolor{green}{TODO: #1}}
\newcommand{\draco}[1]{\textcolor{black}{#1}}

\newcommand{\yl}[1]{\textcolor{black}{#1}}
\newcommand{\fn}[1]{\textcolor{black}{#1}}
\newcommand{\eat}[1]{}

\theoremstyle{plain}
\newtheorem{definition}{Definition}
\newtheorem{example}{Example}
\newtheorem{theorem}{Theorem}

\usepackage{etoolbox,xspace}
\usepackage{algorithmicx}
\usepackage{algorithm}
\usepackage{algpseudocode}
\algtext*{EndWhile}
\algtext*{EndIf}
\algtext*{EndFunction}
\algtext*{EndFor}
\algtext*{EndElse}
\algrenewcommand\algorithmicrequire{\textbf{Input:}}
\algrenewcommand\algorithmicensure{\textbf{Output:}}

\newcommand{\squishlist}{
 \begin{list}{$\bullet$}
  { \setlength{\itemsep}{0pt}
     \setlength{\parsep}{1pt}
     \setlength{\topsep}{1pt}
     \setlength{\partopsep}{0pt}
     \setlength{\leftmargin}{1em}
     \setlength{\labelwidth}{1em}
     \setlength{\labelsep}{0.5em} } }

\newcommand{\squishend}{\end{list}
}

\begin{document}
\title{Sampling over Union of Joins}

\author{Yurong Liu}
\authornote{These authors contributed equally}
\orcid{}
\affiliation{
  \institution{}
  \streetaddress{}
  \city{University of Rochester}
  \state{}
  \country{}
  \postcode{}
}
\email{yliu217@u.rochester.edu}

\author{Yunlong Xu}
\orcid{}
\authornotemark[1]
\affiliation{
  \institution{}
  \streetaddress{}
  \city{University of Rochester}
  \state{}
  \country{}
  \postcode{}
}
\email{yxu103@u.rochester.edu}

\author{Fatemeh Nargesian}
\affiliation{%
  \institution{University of Rochester}
  \streetaddress{}
  \city{}
  \country{}}
  \email{fnargesian@rochester.edu}

\begin{abstract} 
Data scientists often draw on multiple relational  data sources for analysis. A standard assumption in  learning and approximate query answering is that the data is a uniform and independent sample of the underlying distribution. 
To avoid the cost of join and union, given a set of joins, we study the problem of obtaining a random sample from the union of joins without performing the full join and union. 
We present a general framework for random sampling over the  set  union of chain, acyclic, and cyclic joins, with sample uniformity and independence guarantees. We study the novel problem of union of joins size evaluation and propose two approximation  methods based on histograms of columns and random walks on data.  We propose an online union sampling framework that initializes with  cheap-to-calculate parameter approximations   and refines them on the fly during sampling. 
We evaluate our framework on workloads from the TPC-H benchmark and explore the trade-off of the accuracy of union approximation  and sampling efficiency. 
\end{abstract}

\maketitle

\pagestyle{\vldbpagestyle}
\begingroup\small\noindent\raggedright\textbf{PVLDB Reference Format:}\\
\vldbauthors. \vldbtitle. PVLDB, \vldbvolume(\vldbissue): \vldbpages, \vldbyear.\\
\href{https://doi.org/\vldbdoi}{doi:\vldbdoi}
\endgroup
\begingroup
\renewcommand\thefootnote{}\footnote{\noindent
This work is licensed under the Creative Commons BY-NC-ND 4.0 International License. Visit \url{https://creativecommons.org/licenses/by-nc-nd/4.0/} to view a copy of this license. For any use beyond those covered by this license, obtain permission by emailing \href{mailto:info@vldb.org}{info@vldb.org}. Copyright is held by the owner/author(s). Publication rights licensed to the VLDB Endowment. \\
\raggedright Proceedings of the VLDB Endowment, Vol. \vldbvolume, No. \vldbissue\ %
ISSN 2150-8097. \\
\href{https://doi.org/\vldbdoi}{doi:\vldbdoi} \\
}\addtocounter{footnote}{-1}\endgroup

\ifdefempty{\vldbavailabilityurl}{}{
\vspace{.3cm}
\begingroup\small\noindent\raggedright\textbf{PVLDB Artifact Availability:}\\
The source code, data, and/or other artifacts have been made available at \url{https://github.com/DataIntelligenceCrew/sample-union-joins.git}.
\endgroup
}

\section{Introduction}
\label{sec:intro}

Data scientists often draw on multiple sources to collect training data. 
Since most relational sources are not stored as single tables due to normalization, users  often need to perform joins before learning on the join output~\cite{ChenKNP17}. Moreover, data may be collected from distributed sources, each described as a join over internal databases,  data lakes, or web data~\cite{AsudehN22,ChepurkoMZFKK20}. Therefore, the target data is the  union of the  results of joins. 

Joins are expensive  and learning after joins leads to poor training performance due to the introduced 
redundancy avoided by normalization. There have  been efforts to  
enable learning over joins, but they are limited to   certain models, including  linear regression, and cannot be applied to more general models~\cite{KumarNP15,KumarNPZ16,SchleichOC16,ChenKNP17}. 
Fortunately, an important result~\cite{Vapnik:71} from the learning theory 
suggests that learning and approximate query answering~\cite{LiL18}  do not require the full  results and an i.i.d sample can achieve a bounded error. This result holds for any  model. Therefore, in data collection  from multiple  sources, the question to ask is how to obtain a sample from union of sources  without executing join and union. 
Given a collection of joins, the goal is to return a sample of $N$ tuples from the union of the results of joins, independently and at random.  

\begin{figure}[!tbp]
\small
\noindent\begin{minipage}[t]{\columnwidth}
\begin{flalign*}
J\_W&:~
{\bf Customer\_W}(\cdots) \bowtie {\bf Part1\_W}(\cdots) \bowtie {\bf Supplier1\_W}(\cdots) \bowtie &&\\
&{\bf PartSupp1\_W}(\cdots) \bowtie {\bf PartSupp2\_W}(\cdots) \bowtie &&\\
&{\bf LineItem1\_W}(PartKey1,\cdots) \bowtie {\bf Orders1\_W}(OrderKey1, \cdots) \bowtie\\ 
&{\bf Orders2\_W}(OrderKey1, \cdots) \bowtie {\bf LineItem2\_W}(PartKey1,\cdots)\\
\end{flalign*}
\vspace{-7mm}
\begin{flalign*}
\bigcup~J&\_E:~
{\bf Customer\_E}(\cdots) \bowtie {\bf PartSupplier\_E1}(\cdots) \bowtie\\
&{\bf PartSupplier\_E2}(\cdots) \bowtie\\ 
& {\bf DoubleOrders\_E}(OrderKey1, OrderKey2, \cdots) \bowtie &&\\
& {\bf LineItem\_E1}(OrderKey1,\cdots) \bowtie {\bf LineItem\_E2}(OrderKey2,\cdots)\\ 
\end{flalign*}
\vspace{-7mm}
\begin{flalign*}
\bigcup~J&\_MW:~ 
 {\bf Customer\_MW}(\cdots) \bowtie {\bf DoublePartSupplier\_MW}(\cdots) \bowtie\\ 
 & {\bf DoubleOrdersLineItem\_MW}(\cdots)\\
\end{flalign*}
\vspace{-7mm}
    \caption{Example of Union of Joins on Denormalized TPC-H.}
    \label{fig:unionexample}
\end{minipage}
\vspace{-5mm}
\end{figure}

\begin{example} \label{ex:bundlepromo}
Suppose a data scientist in an online retail  company wants to 
train a model for applying a promotion 
to the future bundle orders of customers. 
To do so, the data scientist needs a random and independent sample of size $k$ of customer data and their bundle purchases from the underlying distribution. 
Suppose the customer order data is stored in various databases, each having its own schema. 
For example, the company may have one database for suppliers of each east, west, and midwest region. 
Obtaining customer-bundle data requires constructing a query for each region  database, as shown in Fig.~\ref{fig:unionexample}, then  unioning the results of  queries.  Since there is no single relation that contains all the required features, 
these queries need to join data from  various relations. 
Note that   $J_W$ is a cyclic join and  $J_E$ and $J_{MW}$   are acyclic. In $J_W$, relation   {\tt\small Orders}  is  self-joined to obtain the information of  items in the same order (bundle purchases). All three joins have the same output schema. 
\eat{All three joins have unionable result schemas, i.e., there is a one-to-one mapping between their attributes. The attribute names have been renamed for better readability.} 
To construct the target dataset, 
the first challenge is  
although some of these queries are performed on heavily denormalized relations (or views), for example, {\tt\small PartSupplier} relation in $J_E$, 
since some base relations, for example, the {\tt\small LineItem} and {\tt\small Orders}, are very large, performing a full join becomes very expensive. 
\end{example}

The problem of random sampling over a single join has been actively studied  since the 1990s~\cite{joinsynopses}. 
The goal is to obtain a random and independent sample from join $J$, without performing the full join, such that the probability of each tuple in the sample is $1/|J|$.  
One solution is to join samples of base relations to obtain sample join tuples~\cite{joinsynopses}. 
However, the join of samples produces a much smaller number of join tuples than samples. Moreover, it is shown that the obtained join samples do not guarantee  independence~\cite{HuangYPM19}. 
For approximate query answering, 
some techniques such as RippleJoin~\cite{ripple} and WanderJoin~\cite{2016_wander_join} manage to 
use non-random/independent and random/dependent samples, respectively.    
Other techniques for sampling over join apply the accept/reject sampling paradigm to guarantee i.i.d~\cite{Olken,Chaudhuri:1999}. 
The most recent work by Zhao et al. proposes a framework for  sampling 
over one join that handles general multi-way joins~\cite{2018_sample_join_revisit}. 
The motivation of random sampling over join is tightly connected to join size estimation which  has also  been a point of interest in the database community due to its application to query optimization~\cite{olken1995random,Chaudhuri:1999,agm}.

\begin{example} Continuing with Ex.~\ref{ex:bundlepromo}, the second challenge is to union join samples such that {\em uniformity} is guaranteed, i.e., each tuple has the probability $\frac{1}{|J_W\cup J_E\cup J_{MW}|}$ of being in the final sample. A naive solution 
is to union samples of joins, obtained in an offline manner. 
Suppose we apply an off-the-shelf sampling over join algorithm and 
obtain samples $S_W$, $S_E$, and  $S_{MW}$ from $J_W$, $J_E$, and  $J_{MW}$, respectively. We have $P(t\in S_W)=1/|S_W|$,  $P(t^\prime\in S_E)=1/|S_E|$, and  $P(t^{\prime\prime}\in S_{MW})=1/|S_{MW}|$. It is easy to show that $U = J_W\cup J_E\cup J_{MW}$ does not guarantee uniformity and tuples have unequal probability of appearing in $U$.  
Consider the contradicting example of $r\in S_E, r\notin S_W, r\notin S_{MW}$, 
we have $P(r\in U) = 1/|S_E|$, however, if $r\in S_E\cap S_W\cap S_{MW}$, 
we get $P(r\in U) = (\frac{1}{|J_E|}+\frac{1}{|J_W|}+\frac{1}{|J_{MW}|})\cdot\frac{|S_E\cap S_W\cap S_{MW}|}{|U|}$, because we do set union and keep one instance of overlapping tuples.  
An accept/reject sampling algorithm can help to adjust this probability to obtain $1/|U|$, however, as we show in \S~\ref{sec:setunion}, 
the algorithm  needs to know {\em apriori}, the size of each join and their union, which requires the overlap size of all combinations of  $J_E$, $J_W$, and $J_{MW}$.    
One idea may be to estimate the overlaps and unions from the samples. 
However, that would not be a viable option, since just like joining samples or relations, the probability of obtaining samples from the overlapping regions of joins is low. 
\end{example}

In this paper, we present a generic  framework for random sampling over  the union of joins. 
\fn{In particular,  we consider sampling set union with replacement. Sampling from the disjoint union is a straightforward extension of the  set union. The classic join sampling~\cite{Chaudhuri:1999,olken1995random} and the recently revisited framework~\cite{2018_sample_join_revisit} consider random sampling {\em with replacement} over join.   
Another relevant problem is the random enumeration of the result of the union of acyclic conjunctive queries~\cite{CarmeliZBKS20}. The intermediate results of a random query result enumeration algorithm can be considered as a random sample from the union {\em without replacement} which is different than our problem. Moreover, in this paper, we study union sampling over a larger class of joins (chain, cyclic, and acyclic).  In \S~\ref{sec:complexity}, we provide an elaborate discussion and analytical comparison of this line of work with our framework.}   

There are several challenges to addressing the sampling over the union of joins problem. 
First, unioning random samples from joins does not guarantee uniformity. Our solution is an  accept/reject sampling algorithm that defines  Bernoulli and non-Bernoulli  probability distributions for selecting joins. The latter mimics the behavior of union calculation. Second, it turned out that to guarantee uniformity, the sampling framework  needs to know the size of each join and the size of the union of joins apriori. Although the problem of set union size approximation~\cite{KarpLM89,BringmannF08,AlonMS96} and its online extension to streams~\cite{ChakrabortyMV16, DagumKLR00, KarpLM89, MeelSV17} have been extensively studied in the approximate counting literature,  to the best of our knowledge,  there is no study that addresses the problem of approximating the union  size of joins without performing the full join and overlap. 

Third, \direct estimation requires knowing the  overlap of an exponential number of sets of joins, each set in the powerset of joins. We reduce the space of calculation by reformulating the problem to use smaller-unit statistics, called $k$-overlaps, of each join, which is the size of the subset of a join result that is shared with {\em exactly} $(k-1)$ other joins. 
\fn{Next, we propose two instantiations of the framework for estimating the overlap of joins with an  arbitrary number of relations and  all join types (chain, cyclic, and acyclic): 
a \direct method and a \randomwalk method. The \direct technique is cheap and requires knowing limited  statistics of joins. It may incur  a loose bound, thus, a high rejection rate, under circumstances. The \direct method is highly suitable for data in the wild or scenarios, such as data markets, where limited metadata is available but access to the whole data is infeasible. 
The \randomwalk method  is accurate in estimating parameters and results in low delay. It needs sampling for parameters warm-up  and provides theoretical guarantees.  
To balance the trade-off of parameter estimation cost and sampling efficiency, we propose an \onlineunion sampling algorithm that initializes and updates  parameters with the \direct and \randomwalk methods, respectively, and reuses the samples obtained during \randomwalk while ensuring uniformity.}

\noindent In this paper, we make the following contributions: 
\squishlist
    \item We present the problem of random sampling over the union of joins. 
    \item We design a framework for sampling over the set union of joins of  types chain, cyclic, and acyclic  (\S~\ref{sec:setunion}). Any instantiation of the framework always returns uniform and independent samples from the full  result (Theorem~\ref{th:setunionsampling})  
    but with  different sampling efficiency (\S~\ref{sec:onlineoverlap}). 
    \item We design \direct (\S~\ref{sec:joinoverlap}) and  \randomwalk (\S~\ref{sec:online}) methods to bound the size of overlap of any collection of chain, acyclic, or cyclic joins. 
    \item We present an \onlineunion sampling technique that balances the latency and warm-up cost trade-off (\S~\ref{sec:onlineoverlap}). 
    \item We perform extensive experimental evaluations using the TPC-H benchmark to investigate the error and runtime of parameter estimation and sampling methods(\S~\ref{sec:eval}). We also evaluate the scalability of our framework with respect to relation size, number of samples, and overlap size. 
\squishend

\section{Problem Definition}
\label{sec:problemdef}

Let $\mathcal{A}$ be the universe of attributes and $\mathcal{A}_i$ be the  attributes in relation $J_i$. We are given a set of 
joins $S = \{J_1, \ldots, J_n\}$. A join $J_j$ is defined as $J_j = R_{j,1} \bowtie_{A_{j,1}} R_{j,2} \bowtie_{A_{j,2}} \cdots \bowtie_{A_{j,n_1-1}} R_{j,n_1}$, where  $R_{j,1}, \cdots, R_{j,n_1-1}$ are base relations. 
Similar to relational algebra,  we assume all joins have the same output schema after performing the join in terms of the number and name of attributes. Note that joins can still have different lengths and different  relations. We also  assume that join attributes are standardized to have the same names. We only mention attribute names when needed. In relational algebra, there are two types of unions: set union and disjoint union. The former eliminates duplicate tuples from the result of a union and the latter keeps the duplicates. 
The notion of unionability~\cite{NargesianZPM18} can be applied on base relations to align attributes such that  joins incur the same schema.

The problem of sampling over a union of joins is to return each tuple with probability $1/|union(J_1,\cdots,J_n)|$, 
where union may be set or disjoint union.  
Returning just one sampled tuple is usually not enough,
therefore, we would like to generate totally independent sampled tuples continuously until a certain desired sample size $N$ is reached. 
We formulate the sampling set union and disjoint problems as follows.

\begin{definition} {\em{\bf (Sampling Disjoint Union of Joins)}} Given a set of joins $S = \{J_1, \ldots, J_n\}$, 
return $N$ independent samples from $V = J_1\uplus\ldots\uplus J_n$ such that each sampled tuple is returned with probability $\frac{1}{|V|} = \frac{1}{J_1+\ldots+J_n}$.  
\end{definition}

\fn{Sampling from the disjoint union is straightforward. Given the disjoint union $V = J_1\uplus\ldots\uplus J_n$,  we first select a join  $J_j$ with probability $P(J_j) = \frac{|J_j|}{|J_1+\ldots+J_n|}$, then, we  select a random tuple from $J_j$. This means the probability of each sampled tuple $t$ is $P(t) = \frac{|J_j|}{|V|} \cdot \frac{1}{|J_j|} = \frac{1}{|V|}$. We repeat the process until $N$ sampled tuples are obtained. 
This algorithm always returns independent samples because a returned sample is always uniform regardless of the previous sampling iterations. 
Methods of sampling a tuple from a single join  have long been a popular problem~\cite{Olken,Chaudhuri:1999,2018_sample_join_revisit,viswanath1998join,Chaudhuri:1999}. We revisit random sampling over join in \S~\ref{sec:samplingjoin}. }

The set union operation eliminates duplicate tuples from the result of the union. As such, an i.i.d sampling algorithm  over the set union should return each tuple in the universe of the set union  with the probability of the size of a set union. 

\begin{definition} {\em{\bf (Set Union of Joins Sampling)}} Given a set of joins $S = \{J_1, \ldots, J_n\}$, let $\mathcal{U}$  be the discrete space of unique tuples in $U = J_1\cup\ldots\cup J_n$. Return $N$ independent samples from $\mathcal{U}$, such that each sampled tuple is returned with probability $\frac{1}{\vert J_1\cup\cdots\cup J_n \vert}$. 
\end{definition}

\fn{\section{A Union Sampling Framework}}
\label{sec:setunion}

Let $U$ be the universe of tuples in  the set union of joins. We assume there are no duplicates in each join. 
Given the set union $U = \bigcup_{j = 1}^n J_j$, we want for each value $u\in U$,  $P(t=u)=\frac{1}{\vert U \vert}$. 
\begin{example} Consider joins $J_1$ and $J_2$  that have the same output schema.  Suppose  $t_1 = (3,6,4) \in J_1$ and $t_2 = (3,6,4) \in J_2$. The value 
of each tuple $t$, namely $t.val$,  can be obtained by concatenating its attribute values using a standard convention.  Then, by the definition of a set, $t_1$ and $t_2$ refer to the same tuple, say $u$, in the universe $U = J_1 \cup J_2$. We want $P(u)$, the probability of selecting a tuple with value $u$ from $U$, to be $\frac{1}{|U|}$. Tuples $t_1$ and $t_2$ are distributed in different joins. 
Hence, $u$  is obtained if $t_1$ or $t_2$ are sampled from their corresponding joins. That is, we want $P(t=u) = P(t_1) + P(t_2) = \frac{1}{|U|}$. Note that we may have a sampling with replacement or we may get both $t_1$ and $t_2$ in the sample. Our framework guarantees that $P(t=u) = P(t_1.val) = P(t_2.val) = \frac{1}{|U|}$, whether we choose to remove duplicates or not. 
\end{example}

At each sampling iteration, the framework performs two steps: join selection and join random sampling. The framework continuously samples tuples, with replacement, with $1/|U|$ probability,  until the desired sample size $N$ is reached.  
A straightforward way is based on the union trick~\cite{DurandS11}. At each iteration, we iterate  through all  joins and select a  join  
with the Bernoulli probability $P(J_j) = |J_j|/|U|$. This means multiple joins may be selected in each iteration. 
Upon selecting $J_j$, we randomly sample a tuple $t$ from $J_j$ with replacement.  Recall  $u = t.val$ denotes the value of tuple $t$. 
We accept  tuples with duplicate values $u$, 
only if they are sampled from the same joins, otherwise, we accept the tuples. This means a duplicate tuple $t$ is  retained only if it is sampled from the first join where  $u = t.val$ was observed.  With this description, a tuple value $u\in U$ is returned upon first selecting a join $J_j$ that contains $u$ with probability $|J_j|/|U|$, then sampling $J_j$ with probability $1/|J_j|$. This guarantees that  every  value $u\in U$ is returned with probability $\frac{|J_j|}{|U|}~.~\frac{1}{|J_j|}=\frac{1}{|U|}$. 
 Despite its simplicity, this algorithm has a high rejection ratio for highly overlapping joins and may result in high latency. 
\draco{This is attributed to the utilization of a two-phase framework, which is essential for ensuring uniformity in sampling. Next, we describe a join selection algorithm with a more careful selection of joins. In \S~\ref{sec:onlineunion}, we propose a novel approach that leverages computation performed in the first phase to reduce latency in the second stage. }

\subsection{Non-Bernoulli Join Selection}
\label{sec:nonbernoulli}

\eat{An alternative way to the above approach  is to probabilistically 
select only one join  in each round. 
This way we can obtain exactly $N$ samples and can avoid discarding samples and repeating the last round. 
The goal is to keep samples from an overlap area only when they are sampled  from exactly one predetermined join and discard otherwise.}

\fn{The above technique keeps samples from an overlap area of joins only if they are sampled  from exactly one predetermined join. Consider  two  joins $J_1$ and $J_2$ with overlapping data region $B$ in Fig.~\ref{fig:union_op_bin}. We select and keep any sample $t_1\in J_1$. 
Later, upon selecting $J_2$,  if we sample a $t_2\in B$, we reject $t_2$. The trick to avoiding rejection is to keep the $B$ from $J_1$ as the only space we sample from $J_1$.}  
Therefore, we have  $P(J_1)=\frac{|A+B|}{|A+B+C|}=\frac{|J_1|}{|U|}$, and $P(J_2)=\frac{|C|}{|A+B+C|}=\frac{|J_2|-|B|}{|U|}$. 

Our join selection is outlined in Algorithm~\ref{algNB}. 
Prior to sampling, the algorithm needs to decide which overlapping region is restrictively  sampled from which join. We call this division of joins a {\em cover} (line 2 of Algorithm~\ref{algNB}).    
A cover over joins $S = \{J_1, \cdots, J_n\}$, namely $C = \{J_1^\prime, \cdots, J_n^\prime\}$, is an ordering over $S$ such that $J^\prime_i = \{t\in J_i\vert t\notin\bigcup_{j<i}J^\prime_j\}$. In fact, a cover $J_i^\prime$ of join $J_i$ is  a selection query over join  $J_i$. A cover of $S$ can be created by starting from the first join and keeping or removing overlapping parts.  Fig.~\ref{fig:union_op_gen} illustrates an example of a cover for three overlapping joins. Given a cover $C$, to  calculate the size of $J^\prime_i$,  
we simply follow the inclusion–exclusion principle. Let $O_{\Delta} = \bigcap_{J_j\in\Delta}J_j$ and $S_i$ represent the set of joins that \fn{appear before $J_i$ in the ordering offered by $C$, then we have the following.} 
\begin{equation*}
\label{eq:jprime}
    |J^\prime_i| = |J_i| + \sum_{m = 1}^{i-1}~\sum_{\substack{\Delta \subset S_i, |\Delta| = m}} (-1)^{m} |O_{\Delta} \cup \{J_i\}|
\end{equation*}

Based on this cover, each $J_i$ is selected with  $P(J_i) = \frac{|J_i^\prime|}{|U|}$. When sampling, we should always follow the cover we pre-defined, i.e., for any sample $t\in J_i$, we should discard it if  $t \notin J_i^\prime$.  
\fn{However, if we do not have overlap information apriori, upon selecting $J_i$ and sampling $t$, it is not possible to verify whether $t$ is in  $J^\prime_i$ or not.  Thus, we face a non-trivial case  when we sample $t\in J_i\setminus J_i^\prime$.} 
If we later sample $t$ from $J_j$ \fn{with $J^\prime_j\cap J_i\neq 0$,} i.e., $J_j^\prime$ covers the overlapping part  with $J_i$, 
we should do a critical operation, called {\em revision}. 
This means we remove $t\in J_j$ from the sample 
and re-sample $J_j$, while keeping the $t$ from $J_i$. 

\begin{example} \fn{Consider joins $J_1$, $J_2$, and $J_3$ of Fig.~\ref{fig:union_op_gen}. A cover for these joins are highlighted with blue, red, and green colors. The algorithm selects $J_1$, $J_2$, and $J_3$ with probability $|J^\prime_1|/|U|$, $|J^\prime_2|/|U|$, and $|J^\prime_3|/|U|$, respectively. Suppose at some iteration we  have selected $J_2$ and sampled $t\in J_2\setminus J^\prime_2$. Suppose now we select $J^\prime_1$ and sample the same $t$. Because the cover tells us to sample $J_2$ only from $J^\prime_2$ area, we 
} remove $t$'s from the target set, accept $t$ that's sampled from $J_1$ and assign it to $J_1$ in the record. 
\end{example}


\begin{figure}[!ht]
    \begin{subfigure}[t]{0.32\linewidth}
        	\centering
            \includegraphics[width=\linewidth]{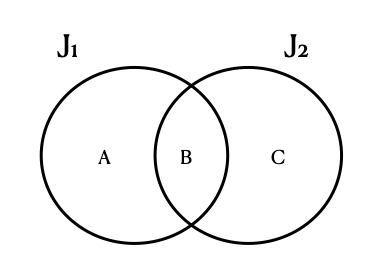}
        	\vspace{-6mm}
        	\caption{}
            \label{fig:union_op_bin}
    \end{subfigure}
    \hfill
    \begin{subfigure}[t]{0.32\linewidth}
        	\centering
            \includegraphics[width=\linewidth]{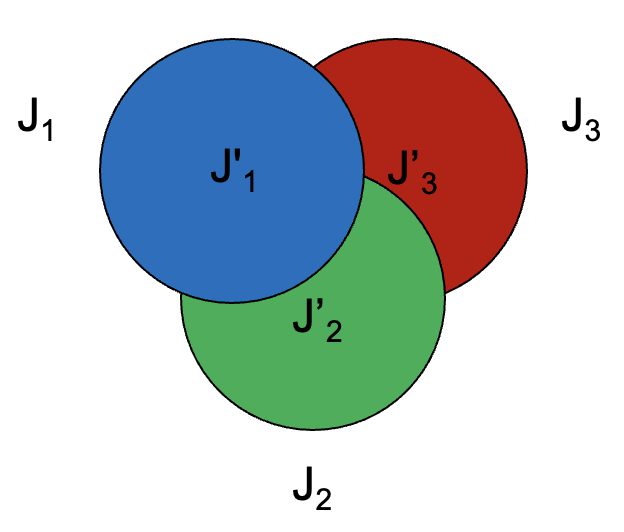}
        	\vspace{-6mm}
        	\caption{}
            \label{fig:union_op_gen}
    \end{subfigure}
    \begin{subfigure}[t]{0.32\linewidth}
    \centering
     \includegraphics[width=\linewidth]{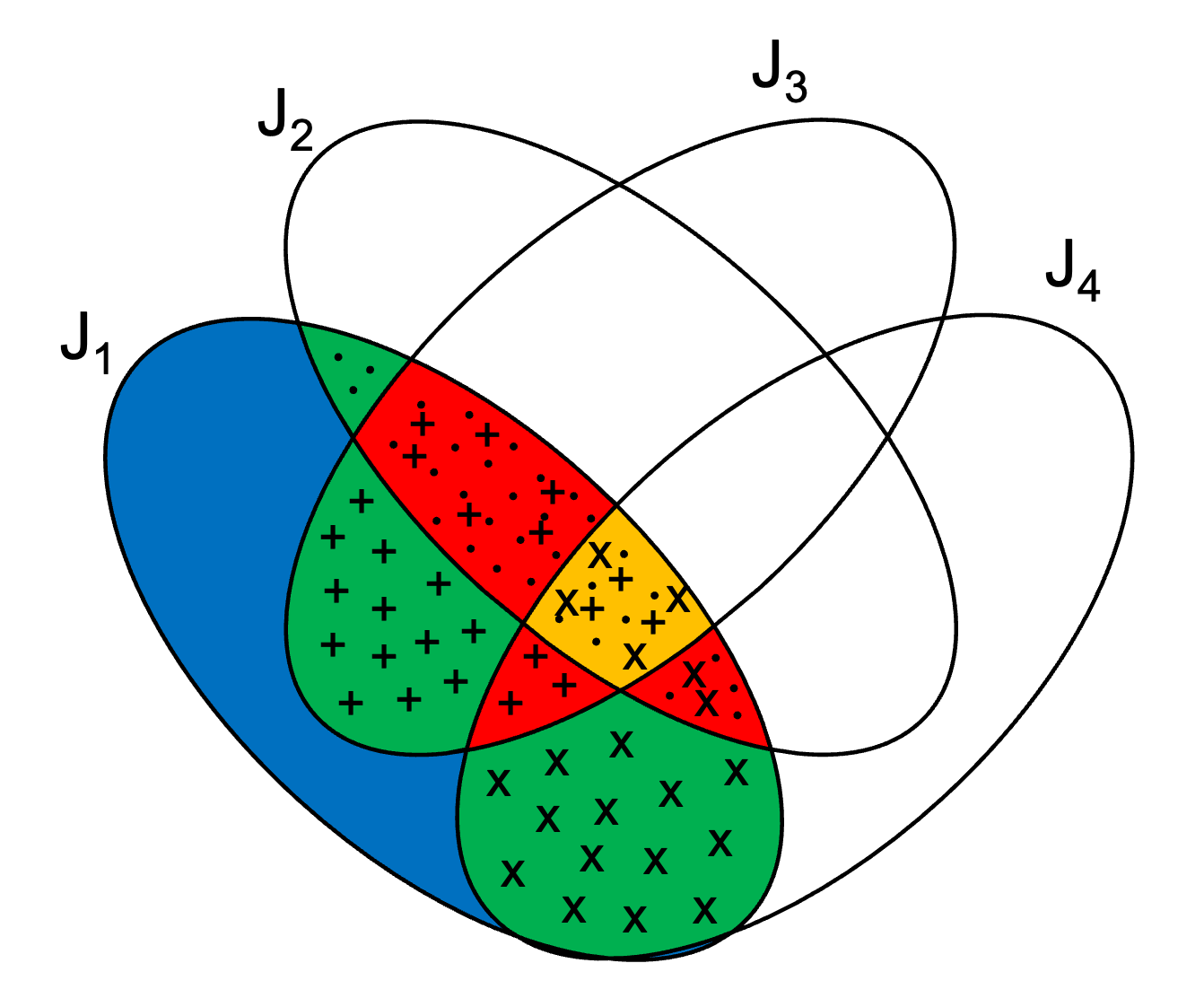}
    \caption{}
    \label{fig:koverlap}
    \end{subfigure}
    \label{fig:union_op}
    \caption{(a) union operation, (b) cover for three joins, and  (c)  $\mathcal{A}_j^k$ of four joins.}
    \vspace{-3mm}
\end{figure}
    
\eat{Note that, in order to define join selection probabilities, both  algorithms need to estimate join size, $|J_i|$, and union size, $|U|$. Our framework computes these statistics during a warm-up phase. 
For join size estimation, we draw on the rich body of work in the literature. While any join size estimation technique can be plugged into our framework, \S~\ref{sec:joinsize} revisits some of the existing techniques. Set union size estimation is a core problem to our paper (\S~\ref{sec: setunionsize}). To the best of our knowledge, we are the first to consider the problem of estimating the union of joins. Finally, upon selecting a join, the algorithm needs to obtain a random sample. Our framework incorporates the state-of-the-art algorithm on random sampling over join~\cite{2018_sample_join_revisit}, which we overview in \S~\ref{sec:samplingjoin}.} 
\begin{theorem}\label{th:setunionsampling}
Given joins $S=\{J_1,\ldots, J_n\}$, Algorithm~\ref{algNB} returns each result tuple $t$ with value $u$ with probability 
$\frac{1}{|J_1\cup\ldots\cup J_n|}$. 
\end{theorem}
\begin{proof} Intuitively, a cover defined by Algorithm~\ref{algNB} decides from which join exclusively a value in the overlap of a collection of joins is sampled.  Recall $U$ is the universe of the set union of tuples of joins, i.e., $\{u | u \in \cup_i J_i\}$. Algorithm~\ref{algNB} uses a mapping strategy function $f: U \rightarrow S$  that tells us to which $J_i$ a specific $u$ is assigned.  
Note that $u$ could belong to multiple $J_i$'s, however, $f$ refers to the unique $J_i$ from which $u$ can be sampled.  
Let a cover $C$ of $S$  be the quotient space of $U$ over $f$ and $g:S\rightarrow C$ be a mapping function such that $g(J_i)=J_i^\prime$.  
Then, $g \circ f $ will map each $u\in U$ to a join in cover $C$. 
For all $u$, we denote $|g(f(u))|$ to be $|\{u^\prime | g(f(u^\prime))=g(f(u))\}|$. In other words, the probability of sampling a $u\in U$ depends on the probability of selecting $g(f(u))$ followed by sampling $u$ from $g(f(u))$. 
 Therefore, we obtain the probability of $P(t = u)$ as follows. 
\begin{align*}
    P(t = u) = P(f(u))\cdot \frac{1}{|g(f(u))|} = \frac{|g(f(u))|}{|U|} \cdot \frac{1}{|g(f(u))|} = \frac{1}{|U|}
\end{align*}
\end{proof}

Computing the probability distribution of line 2 of Algorithm~\ref{algNB} requires the knowledge of $|J^\prime_i|$ as well as $|U|$. In \S~\ref{sec:setunionsize}, we describe ways of estimating the overlap of $k$ joins and $|J^\prime_i|$. 

\begin{algorithm}
\caption{\fn{Union Sampling}}
\begin{small}
\begin{algorithmic}[1]
\Require{Joins $S = \{J_j, 1\ \leq j \leq n\}$, tuple count $N$}
\Ensure{Tuples $\{t_i, 1 \leq i \leq N\}$}
\State $\{|J_j|, 1\ \leq j \leq n\}, |U| \leftarrow warmup(S)$
\Comment{\fn{\direct (\S~\ref{sec:joinoverlap}) or \randomwalk 
(\S~\ref{sec:online})}}

\State $\{|J_j'|\} \leftarrow\fn{cover(S)}$ 
\State $T \leftarrow \{\}$
\Comment{target sample}
\State $orig\_join_i \leftarrow \{\}$ \Comment{record of original join of seen tuples}
\While{$n < N$}
\State{select $J_j$ with probability $\frac{|J_j'|}{|U|}$}
 \State{$t \leftarrow$ a random sample from $J_j$}

\If{$t \in orig\_join_i $ for any $i < j$} reject $t$
\Else
\If{$t \in orig\_join_i $ for any $i > j$} 
\Comment{\fn{revision}}
\State remove $t$ from $orig\_join_i$ \fn{ and add $t$ to $orig\_join_j$}
\State remove all $t$'s from $T$
\EndIf
\If{$t \notin orig\_join_i$} add $t$ to $orig\_join_j$
\EndIf
\State $T \leftarrow T \cup \{t\}$
\EndIf
\EndWhile
\State \textbf{return} $T$
\end{algorithmic}
\label{algNB}
\end{small}
\end{algorithm}

\subsection{Join Sampling Revisited} 
\label{sec:samplingjoin}

To sample a single join (line 7 of Algorithm~\ref{algNB}), we consider the work by Zhao et al.~\cite{2018_sample_join_revisit}, which is a generic framework for sampling from any type of join. 
The framework defines a join data graph where each tuple in a relation is a node. Each tuple $t$ is labeled with a  weight defined as the upper bound for the number of  tuples in the join result that $t$ yields. 
The framework performs accept/reject sampling. Each tuple from a relation is sampled with some probability based on its weight and is rejected with some rate in terms of the weights to guarantee uniformity. 
We make some design choices to adopt the join sample framework of Zhao et al. as a subroutine in our union sampling framework. 

First, for weight instantiation, we use three  techniques: extended Olken's, exact, and Wander Join~\cite{2016_wander_join}, proposed by Zhao et al.~\cite{2018_sample_join_revisit}.  
Second, this framework requires index structures over base relations to know which tuples can be joined together. Instead,  we use hash tables for relations to maintain tuples' joinability  information. 
Third, one limitation of Zhao et al.'s framework is the assumption of having only key-foreign key joins between relations. Since in a generic join, some tuples may not have a joinable tuple in other relations, we  release this assumption by modifying the Extended Olken's to 
set the weights (and hence probabilities) of those tuples to zero with an extra linear search in the hash tables. 

Finally, to obtain the accept/reject ratio, this framework   
allows us to plug in any of the join size upper-bound  estimations. We also need to compute the size upper bound of joins in Algorithm~\ref{algNB}.  To do so, in what follows, we adopt parts of the algorithm proposed in \citet{ngo2018worst} and  extend Olken's algorithm~\cite{Olken} to calculate the upper bound on the size of joins of an arbitrary number of relations.  

Assume a join $J = R_1 \bowtie_{A_1} R_{2} \bowtie_{A_2} \cdots \bowtie_{A_{n-1}} R_{n}$. Let $M_{A_i}(R_{i+1})$ be the maximum value frequency in  attribute $A_i$ of relation  $R_{i+1}$. Since each tuple in $R_2$ with value $v$ for $A_i$  can be matched with maximum $M_{A_i}(R_{i+1})$ tuples of $R_{i+1}$ on $A_i$, we have the following upper bound for the size of a join $J$: 
$|J| \leq |R_{1}| \cdot \prod_{i = 1}^{n-1} M_{A_i}(R_{i+1})$. In our framework, we consider the above extension of Olken's algorithm for join size estimation in all  algorithms. 

\fn{\subsection{Cost Analysis}} 
\label{sec:complexity}

Since the subroutine of sampling from a join in Algorithm~\ref{algNB} is based on the existing algorithms, for the cost analysis, we decouple the delay of random sampling over join from our algorithm and consider the total number of samples obtained from the join subroutine as our total cost. 

\begin{theorem}\label{th:time}
Given joins $S=\{J_1,\ldots, J_n\}$, the expected total sampling cost of Algorithm~\ref{algNB} for returning $N$ uniform and independent samples is $N + N \log N$.  
\end{theorem}
\begin{proof} Given a cover \yl{$C=\{J_j^\prime \;\vert\; j \in [1,n] \cap \mathbf{Z}\}$}, Algorithm~\ref{algNB} samples each join $J_j$ with probability $|J^\prime_j|/|U|$. Let $N_j$ be the number of tuples from $J_j$ that are in the final result. Based on Algorithm~\ref{algNB}, we know a tuple from $J_j$ is in the final sample if it is obtained from $J^\prime_j$. Therefore, we have $N_j = \frac{|J^\prime_j|}{|U|}~.~N$, in expectation. 
Let $\psi_j$ be the number of tuples Algorithm~\ref{algNB} ever obtains from $J_j$. A tuple may be a rejected, accepted, or revised sample, because the set of tuples from different joins may intersect. Based on the union bound, the number of iterations of Algorithm~\ref{algNB} is bounded by the sum of the number of tuples sampled from each join. Using this principle, we have the  expected total number of iterations of $\psi \leq \sum_{j=1}^n\psi_j$. Given $N_j$ coupons, the coupon collector's problem provides a bound for the number of samples we expect we need to draw with replacement before having drawn each coupon at least once~\cite{MotwaniR95}. 
This result allows us to obtain the expected value of $\psi_j=N_j\log N_j$. Therefore, we have the following expected number of iterations. 
$$\psi \leq \sum_{j=1}^{n}N_j\log N_j = \sum_{j=1}^{n}N.\frac{|J^\prime_j|}{|U|}\log \left( N.\frac{|J^\prime_j|}{|U|}\right)$$
Let $\alpha_j = \frac{|J^\prime_j|}{|U|}$. We have the following. 
\begin{align*}
    \psi \leq \sum_{j=1}^{n}\alpha_j.N \log (\alpha_j.N) =  &N\left(\sum_{j=1}^n\alpha_j\log\alpha_j + \sum_{j=1}^n\alpha_j\log N\right) 
\end{align*}
From the definition of cover, we  know $\sum_{j=1}^n\frac{|J^\prime_j|}{|U|}=1$. Therefore, we have the following bound on the expected total time. 
$$\psi \leq N(\log(H(n)) + \log N) \leq N + N \log N$$
\end{proof}

We remark that  although our algorithm does not have a
strict and deterministic guarantee on the delay between samples, our total time is on par with the $\mathcal{O}(N \log N)$ time of the algorithm proposed by Carmeli et al., for the random enumeration of the result of the union of conjunctive queries, where $N$ is the number of answers~\cite{CarmeliZBKS20}. 

\section{Size of Set Union of Joins}
\label{sec:setunionsize}

Executing full joins and computing set union  is costly. We propose a novel way  of computing the set union size  by using the size of joins and the size of the overlap of joins.  To do so, we first separate each join  $J_j$ into $n$ disjoint parts, denoted as $J_j = \bigcup_{k = 1}^n \mathcal{A}_j^k$, where $\mathcal{A}_{j}^k$ is the set  of tuples of $k$-th overlap in $J_j$, i.e., each tuple in $\mathcal{A}_j^k$ belongs to $J_j$ and appears in exactly $k-1$ other joins. The base case $\mathcal{A}_{j}^1$ includes the tuples in $J_j$ that are  the set complement of all overlaps.  Fig.~\ref{fig:koverlap} represents the $\mathcal{A}_{j}^k$ areas for a join $J_1$. 
\eat{\begin{example}
Fig.~\ref{fig:koverlap} represents the $\mathcal{A}_{j}^k$ areas for join $J_1$, where the red area is $\mathcal{A}_{j}^1$, the blue area is $\mathcal{A}_{j}^2$, and the green area is $\mathcal{A}_{j}^3$. 
\end{example}}
Since for each $J_j$,  $\mathcal{A}_{j}^k$'s are disjoint, we can define the size of the set union $U$ as follows. 
\begin{equation}
\label{eq:Ajk}
|U| = \sum_{j=1}^n\sum_{k=1}^n\frac{1}{k}|\mathcal{A}_j^k|
\end{equation}

\eat{
\begin{example} \fn{Consider join paths of Fig.~\ref{fig:ajk}. The disjoint union of these joins is the sum of $|\mathcal{A}_j^k|$'s of all joins. To compute the set union we need to exclude from this sum the regions that have been counted more than once. For example, $\mathcal{A}_1^2$ is counted twice (in  $\mathcal{A}_2^2\cap J_1$ as well as $\mathcal{A}_3^2\cap J_1$) and $\mathcal{A}_1^3$ is counted thrice (in all join paths). 
To avoid over-counting, the $\mathcal{A}_j^k$'s are weighed by $1/k$, in Eq.~\ref{eq:Ajk}.}
\end{example}
}

Note that $\mathcal{A}_j^k$ is non-trivial information, which requires combining the overlap size of $k$-combinations of joins. There are two challenges for computing $\mathcal{A}_j^k$. First,  there is no relationship between the pairwise overlap information and higher order $k$-th overlap, $\mathcal{A}_j^k (k > 2)$. Second, computing a pairwise overlap size without a full join is more challenging than computing a single join size. 

Suppose we have a way of computing the overlap for any set of joins. More formally, given a collection $\Delta\in S$ of joins, $\mathcal{O}_{\Delta}$ denotes the overlap of joins in $\Delta$.  In \S~\ref{sec:joinoverlap} and~\ref{sec:others}, we describe various algorithms for overlap estimation of all join types (chain, cyclic, and acyclic). 
Now, we turn our attention to computing  $\mathcal{A}_j^k$ using $\mathcal{O}_{\Delta}$. We describe the intuition of our solution with an example. 

\begin{example} \label{ex:koverlap} 
Consider the joins $S=\{J_1, \cdots, J_4\}$ of Fig.~\ref{fig:koverlap}. The areas $\mathcal{A}_1^k$ for $k\in[1,4]$ are color-coded. 
We would like to compute the size of $\mathcal{A}_1^2$. The dotted, $+$, and $\times$ areas included all pairwise overlaps. Suppose we first compute the sum of  the pairwise overlap size of joins with $J_1$, i.e., $\sum_{\Delta\in \mathbb{P}_2\wedge J_1\in\Delta}|\mathcal{O}_{\Delta}|$, where $\mathbb{P}_2$ is the collection of all subsets of size $2$ of $S$.  
However, to determine the area of the overlap of exactly one join  with $J_1$, $\mathcal{A}_1^2$, we need to exclude all $\mathcal{A}_1^3$ and $\mathcal{A}_1^4$ areas. In fact, each subarea of $\mathcal{A}_1^3$ counts twice in the above sum. 
For example, $J_1\cap J_2\cap J_3$ is in both $J_1\cap J_2$ and $J_1\cap J_3$. Similarly, $\mathcal{A}_1^4$ counts three times in the sum of $\mathcal{O}_{\Delta}$'s since it is included in $J_1\cap J_2\cap J_3$, $J_1\cap J_2\cap J_4$, and $J_1\cap J_3\cap J_4$. To avoid over-counting, the $\mathcal{A}_j^k$'s are weighed by $1/k$, in Eq.~\ref{eq:Ajk}. 
\end{example}

\begin{theorem}
\label{Ajk}
Let $S = \{J_1, J_2, \dots J_n\}$ and $\mathbb{P}_k$ be all subsets of size $k$ of $S$,  
then for any join path $J_j$, and for any $1\leq k\leq n$, we have 
$$|\mathcal{A}_j^k| = \sum_{\Delta\in \mathbb{P}_k\wedge J_j\in\Delta} |\mathcal{O}_{\Delta}| - ( \sum_{r=k+1}^{n} \binom{r-1}{k-1} \cdot |\mathcal{A}_j^{r}|).$$
For $k = n$, we have $|\mathcal{A}_j^n| = |\mathcal{O}_{S}|$. For $k = 1$, we have the following.  $$|\mathcal{A}_j^1| = \sum_{\Delta\in \mathbb{P}_1\wedge J_j \in\Delta} |\mathcal{O}_{\Delta}| - \sum_{r = 2}^{n} \binom{r-1}{0} |\mathcal{A}_j^k| = |J_j| - \sum_{r = 2}^n |\mathcal{A}_j^r|$$
\end{theorem} 

\begin{proof}
When $k = n$, $\mathbb{P}_n$ is the set  representing the universe $S$ including $J_j$. Therefore, it is trivial that  
$|\mathcal{A}_j^n| = O_S$, which can be evaluated 
with $\bigcap_{J_j \in S} J_j$. Then, for $k \in [2, n-1] \cap \mathbf{Z}$, we calculate $|\mathcal{A}_j^k|$ dynamically. 
Now, suppose we know $|\mathcal{A}_j^{k+1}|$. 
Recall $\mathcal{A}_j^k$ consists of all tuples in $J_j$ that appear in exactly $k-1$ other join paths. That is, tuples in $J_j$ that are in some $\Delta\in \mathbb{P}_k$ but are not in any higher order overlap  $\Delta^\prime\in \mathbb{P}_{r}$, where $r\in[k+1,n]$. 
Therefore, we first add up all the $k$-th overlap for sets $\Delta\in \mathbb{P}_k$, where $J_j \in \Delta$.
Since $J_j$ is confirmed, we have $\binom{n-1}{k-1}$ number of such sets $\Delta$. 
Note that a tuple $t\in\mathcal{A}_j^k$ may appear in multiple $\Delta\in \mathbb{P}_r, r\in[k+1,n]$. 
Therefore, to get the exact value of $|\mathcal{A}_j^k|$, 
for each $r\in[k+1,n]$, 
we need to count the number of $\Delta\in \mathbb{P}_r$ where $J_j\in\Delta$. 
Starting with $r = k+1$, each such combination of $\Delta\in \mathbb{P}_{k+1}$ contains $J_j$, therefore, it  appears once in remaining $\binom{k}{k-1}$ number of $\Delta^\prime\in \mathbb{P}_k$'s. 
Hence, we need to deduct $(k-1)\cdot |\mathcal{A}_j^{k+1}|$ from the sum. 
Now for the general case $r$, where $k < r\leq n$, after $J_j$ is confirmed, each combination of $\Delta\in \mathbb{P}_r$ has its other $k-1$ paths chosen in $\binom{r-1}{k-1}$ number of $\Delta^\prime\in \mathbb{P}_{k}$, so a total number of $\binom{r-1}{k-1}$ $|\mathcal{A}_j^{r}|$ needs to be deducted from the sum for each $r$. 
Therefore, we can organize the formula of calculating $|A_j^k|$ as shown in the theorem. 
\end{proof}

Using this theorem to compute  $|\mathcal{A}_j^k|$'s for a given $J_j$ and all $k\in[1,n]$, we start by initializing $|\mathcal{A}_j^n|$ with $|\mathcal{O}_S|$ using the method proposed in \S~\ref{sec:joinoverlap}. Then, $|\mathcal{A}_j^{n-1}|$ requires evaluating $|\mathcal{A}_j^n|$ that have been already computed as well as $|\mathcal{O}_{\Delta}|$ for each subset of size $n-1$ of $S$. 
Again, \S~\ref{sec:joinoverlap} is used to compute a $|\mathcal{O}_{\Delta}|$. In general, iterating from $n-1$ to $1$, each $|\mathcal{A}_j^k|$ can be computed from $|\mathcal{A}_j^r|$'s, where $r\in(k,n]$, that have been already evaluated and $|\mathcal{O}_{\Delta}|$'s that can be computed from our method for the pairwise join path overlap. 

Computing the size of a set union  requires computing   the overlap of all $k$-subsets of joins, which is exponential in the number of input joins. 
We remark that in practice the number of input joins is small. 
However, when  $S$ is large, if we compute  $|\mathcal{O}_{\Delta}|$'s in the order of the bottom-up traversal of the powerset lattice of $S$, we can speed up by  reusing some of the computation. 

{\bf \fn{Warm-up Phase:}} Note that computing the exact values of $k$-overlaps and overlaps for an arbitrary number of joins and relations is computationally expensive or infeasible. 
Next, we present two instantiations of the framework for approximating these parameters.  We consider two cases: centralized and decentralized~\cite{HuangYPM19}.  
In a centralized setting, relations are accessible through  direct access to data, such as relations within  databases. We propose \randomwalk for this setting. In a decentralized setting, data is private or expensive to sample. Examples include data markets or large relations in databases. Our \direct method is suitable for this setting. 
\fn{Different instantiations of the framework only differ
in how the union size  bound, and join overlap  bounds are computed during the warm-up phase. 
We remark that both methods guarantee uniformity. 
There is a tradeoff between efficiency and cost of estimation: tighter
upper bounds are more costly to set up, but once in place, can
generate samples more efficiently. 
On the other hand, looser upper
bounds are easier to compute but lead to low sampling efficiency
(due to potentially higher rejection rates).  
We propose a modified version of union sampling based on the \randomwalk method that does not require warm-up and strikes a better tradeoff between upper-bound computation and sampling efficiency.} 

\eat{
\fn{In the next section, we describe two instantiation methods. 
The direct method leverages minimal statistics of base relations, such as maximum degree, and has almost zero setup cost but very low sampling efficiency. The second method requires some sampling cost during the warm-up phase but yields a better upper-bound computation.  
Next, we describe a modified version of union sampling based on the online method that strikes a better tradeoff between upper-bound computation and sampling efficiency. 
Finally, we discuss how to extend to other types of joins in \S~\ref{sec:others}.}
}

\fn{\section{Instantiation with Histograms}}
\label{sec:joinoverlap}

\fn{Database management systems often maintain histograms as a special type of column statistic that provides  more detailed information about the data distribution in a table column during query optimization. These histograms are useful for  cardinality estimation, particularly if the data in a column is skewed. 
In this section, we present ways of estimating join overlap and union size using these histograms and even  more minimalistic statistics such as maximum degrees of tuples in relations. Here, we propose a solution for the case of chain join, inspired by Olken's seminal work on join size estimation\cite{olken1995random}. In \S~\ref{sec:cyclic}, we extend our framework to more generic cyclic and acyclic joins.} 


\subsection{Overlap of Equi-length Chain Joins}
\label{sec:overlapmulti}

\fn{We start with estimating the overlap of multiple chain joins. Suppose all joins consist of the same number of relations and there is a one-to-one mapping between relations of each pair of joins such that mapped relations have the same schema.}  Given a collection of joins $S$ and a subset $\Delta\subseteq S$, let $O_{\Delta} = \bigcap_{J_j\in\Delta}J_j$ be the set of tuples that appear in all $J_j \in \Delta$. Trivially, a loose upper bound for the overlap is $\min\{|J_j|: J_j\in\Delta\}$. We first partition the joins on relations consistently. At each step, we estimate the \fn{overlap} size of each sub-join \fn{dynamically from the overlap of smaller sub-joins} by multiplying \fn{the overlap size of a smaller sub-joins} by the minimum of the maximum degree of values join attributes. 
\fn{For example, for joins of three relations, we first evaluate the overlap of the first relations in all joins. Then, we evaluate the overlap of the first two relations in all joins by multiplying the overlap of the first relations by the minimum of the maximum degree of values of the first relations, and so on.} 

More formally, let $\mathcal{K}(i)$ be the upper bound of the number of overlapping tuples after the $i$-th join. Hence, $|\mathcal{O}_{\Delta}| \leq \mathcal{K}(n-1)$. Let $M_{A_l}(R_{j,i})$ be the maximum degree of values in the domain of a join attribute $A_l$ of relation $R_{j,i}$ of join  $J_j$ and let $d_{A_l}(v,R_{j,i})$ be the degree of value $v$ in the domain of $A_l$. \fn{Note that the statistics of the degree of values are available from the histograms on join attributes.}  We can obtain an upper bound dynamically as $\mathcal{K}(i) = \mathcal{K}(i-1) \cdot \min_{J_j\in\Delta}\{ M_{A_i}(R_{j,i+1})\}$. Note that for $\mathcal{K}_1$ we calculate the bounds based on values, i.e., $ \mathcal{K}(1) = \sum_{v\in\mathcal{C}}\min_{J_j\in\Delta}\{d_{A_1}(v,R_{j,1}) \cdot  d_{A_1}(v,R_{j,2})\}$. 
\fn{So far, this bound requires the full histogram of the first relations in all joins and the maximum degree of values in the remaining relations. If the histograms are available for all join attributes in the relations, we can further refine the bound by replacing the term of the minimum of maximum degrees, $M_{A_i}(R_{j,i+1})$, with the minimum of the average degree of values in the join attributes.}

\subsection{Overlap of Chain Joins}
\label{sec:overlapgen}

We now release this assumption to accommodate joins with arbitrary length  and arbitrary relation schemas. Note that the joins themselves should still have the same schemas after joining. We introduce the {\em splitting method}  that \fn{aims to reorganize joins into joins on relations of the same size, so that the results of \S~\ref{sec:overlapmulti} can be applied.} The splitting method derives new joins by breaking down relations into sub-relations, each sub-relation consisting of exactly two attributes. The derived  joins  have the same schema and are lossless, i.e., each generates the same data as the original join, and all contain the same number of relations.  
Moreover, for each relation in a derived join, there are corresponding relations in other joins. Since the derived joins satisfy the requirements of \S~\ref{sec:overlapmulti} and generate the same data, we can directly apply \S~\ref{sec:overlapmulti}  to estimate the overlap size of the original joins.    
Although the input joins may not include relations with the same schemas, they definitely have  corresponding attributes and the same schema after joining. As such, breaking all relations in sub-relations of two attributes and redefining joins incurs join  with the same number of same-schema   relations. 

Note that our splitting method is different than the normalization in the database theory which aims to decompose relations into sub-relations based on functional dependencies to avoid anomalies~\cite{Codd71a}. {\em Split} relations keep a record of their original sizes for the estimation steps. We call the join between two relations split from the same original relation {\em fake join}. \fn{The following theorem describes a generic way of bounding the overlap of chain joins.} 
\begin{theorem}
\label{direct_olp}
Given a collection of split joins $S$ and a subset $\Delta\subset S$, let $O_{\Delta} = \bigcap_{J_j\in\Delta}J_j$.  Let $M_{A_l}(R_{j,i})$ be  the maximum degree of values in the domain of a  join attribute $A_l$ of relation $R_{j,i}$ of join $J_j$ and let  $d_{A_l}(v,R_{j,i})$ be the degree of value $v$ in the domain of $A_l$. We define the following. 
\begin{align*}
    M_{j,i} = \left\{\begin{array}{l}
    M_{A_i}(R_{j,i+1})\; \text{ if }\;  R_{j,i} \bowtie R_{j,i+1}\\
    \\
    \eat{|R_{j,i}|} 1\; \text{ if }\; R_{j,i} \bowtie' R_{j,i+1}
\end{array}\right.
\end{align*}
Let $\mathcal{K}(i)$ be the upper bound of the number of overlapping tuples after the $i$-th join and let $d_{A_l}(v,R_{j,i})$ be the degree of value $v$ in the domain of $A_l$. We then obtain an upper bound for the overlap size of joins in $\Delta$,  $|O_{\Delta}|$, dynamically as follows. 
\begin{align*}
    |\mathcal{O}_{\Delta}| &\leq \mathcal{K}(n-1) = \mathcal{K}(n-2) \cdot
    \min_{J_j\in\Delta}\{M_{j,n}\}\\
    \mathcal{K}(1) &= \sum_{v\in\mathcal{C}}\min_{J_j\in\Delta}\{d_{A_1}(v,R_{j,1}) \cdot  d_{A_1}(v,R_{j,2})\}\\
    \mathcal{K}(i) &= \mathcal{K}(i-1) \cdot \min_{J_j\in\Delta}\{ M_{j,i}\}
\end{align*}
\end{theorem}

\begin{proof}
The proof of this theorem follows from \S~\ref{sec:overlapmulti} and \S~\ref{sec:overlapgen}.
\end{proof}

\fn{We remark that Theorem~\ref{direct_olp} can become a biased estimator of join overlap if the data is skewed. Here, we present a solution with the least statistics available. We can extend the theorem, to become an unbiased estimator, in a straightforward way to use the histogram information of all join attributes and compute the expected value and upper bound of overlap.}

\section{Instantiation with Random Walks} 
\label{sec:online}

The techniques proposed in \S~\ref{sec:setunionsize}  perform join union size estimation in a direct manner. In this section, we consider an alternative and more accurate way of  estimating join overlap size in an online manner. The idea is to update the join size and overlap size on the fly, during the warm-up phase, by obtaining tuples from join paths and reusing these tuples during the main sampling step.

\subsection{Join Size Estimation Revisited}
\label{sec:onlinejoinsize}

To solve the online aggregation problem over join, wander join proposes an algorithm by performing random walks over the underlying join data graph~\cite{2016_wander_join}. This solution can be applied to join size estimation by computing the \texttt{COUNT} operation over the join.  
A join data graph models the join relationships among the tuples as a
graph, where nodes are tuples and there is an edge between two tuples if they can join. Using a join graph, we can easily obtain successfully joined tuples by performing random walks. The probability of a tuple sampled from a join can be  computed on the fly using the join graph. 
Given a join $J = R_1\bowtie R_2\bowtie\ldots\bowtie R_m$, the probability of a result tuple  $t = t_1\bowtie t_2\bowtie\ldots \bowtie t_m$ is computed as $p(t) = \frac{1}{|R_1|}\cdot\frac{1}{|d_{2}(t_1)|}\cdot\cdots\cdot\frac{1}{|d_{m}(t_{m-1})|}$, where $d_{i}(t_{i-1})$ is the number of tuples in $R_i$ than join with $t_{i-1}$. 
\begin{example} Consider the index graph of $J$ in Fig.~\ref{fig:indexgraph}. The probability of choosing $a_1$ is $\frac{1}{5}$. Then among the three joinable tuples with $a_1$, the probability of selecting $b_2$ is $\frac{1}{2}$. Similarly, the probability of selecting $c_1$  is $\frac{1}{3}$. Therefore, the probability of obtaining tuple $a_1 \bowtie b_2 \bowtie c_1$ is 
$p(a_1 \bowtie b_2 \bowtie c_1) = \frac{1}{5} \times \frac{1}{2} \times \frac{1}{3}$. 
\end{example}


Suppose we have obtained a sample $S$ of size $m$ from a join path $J$. 
Following Horvitz-Thompson estimator~\cite{horvitz1952gsw}, the estimated join size of $J$ based on sample $S$, namely $|J|_S$  can be evaluated as  
$|J|_S = \sum_{t \in S} \frac{1}{p(t_k)}\cdot\frac{1}{m}$~\cite{2016_wander_join}. 
We can update this estimation in real-time as new join samples are obtained. 
Suppose  a new tuple $t_0$ is added to $S$, we can update the join size estimation as follows. 
\begin{align*}
    |J|_{S\cup t_0} &=  \frac{\sum_{t \in S} \frac{1}{p(t_k)} + \frac{1}{p(t_0)}}{(m+1)}
    = \frac{\sum_{t \in S} \frac{1}{p(t_k)}}{m} + \frac{\frac{m}{p(t_0)} - \sum_{t \in S} \frac{1}{p(t_k)}}{(m+1)m}\\
    &= |J|_S + \frac{1}{m+1} \left( \frac{1}{p(t_0)} - |J|_S\right)
\end{align*}
We revisit the mean and variance of $|J|$ later in the discussion of random walk  overlap. 
Hence, a real-time approximate answer is returned with some confidence level, and the accuracy improves as the sample size grows larger. 
Extending from wander join, we have two methods to estimate the overlap sizes.

First, we set an $\alpha$ as a parameter, which is the confidence level we want to achieve. There is a confidence level value $z_\alpha$ corresponding to the $\alpha$. The half-width of the confidence interval is $ \frac{z_\alpha \cdot \sigma}{\sqrt{n}}$, where $n$ is the sample size and $\sigma$, is the standard deviation of the sample set. We terminate the sampling when the half-width becomes less than the threshold we defined.

\subsection{Overlap of Joins}
\label{sec:onlineoverlap}

We described an  algorithm based on random walks for sampling a join and estimating a join size. Given a set $\Delta\in S$ of join paths, we would like to estimate the overlap of joins in $\Delta$, namely $\mathcal{O}_{\Delta}$. 
Let $S_j = \{t_1, t_2, \dots, t_m\}$ denote a collection of sampled tuples from join $J_j\in\Delta$. Let $count(t)$ be the number of occurrences of tuple $t$ in a set. We define  $S_j^\prime$ such that for each tuple $t$ in $S_j$, $S_j^\prime$ contains exactly $\frac{1}{p(t)}$ number of such tuple $t$, i.e.,  
    $S_j^\prime = \{ t \in S_j \mid count(t)=\frac{1}{p(t)} \}$. 
Thus, sample $S_j^\prime$ preserves the distribution of $J_j$. 
We assume uniformity, over overlap, and non-overlap regions among join paths, that is we sample tuples and estimate join sizes by performing random walks,  for any  $J_j\in\Delta$, we have 
$\frac{|\mathcal{O}_\Delta|}{|J_j|} = \frac{|\bigcap_{J_j\in\Delta}S_j^\prime|}{|S_j^\prime|}$. Therefore, a join overlap size is estimated on the fly as follows. 

\begin{equation}\label{eq:onlineoverlap}
|\mathcal{O}_\Delta| = |\bigcap_{J_j\in\Delta}J_j| = |J_j|\cdot   \frac{|\bigcap_{J_i\in\Delta}S_i^\prime|}{|S_j^\prime|}
\end{equation}
How to get the $|\bigcap_{J_i\in\Delta}S_i^\prime|$? We fix a $J_j\in\Delta$ and continually sample from this single source, forming the $S_j$. In each round, if we accept the sample $t$, then we check every$J_i\in\Delta$, where $i \neq j$ to see where $t$ is contained in $J_i$. Since we already have the index for each $J_i$(stored in hash tables), this operation could be cheap since it just requires $(N-1)\times (M-1)$ queries with key, where $N=|\Delta|$ and $M$ is the number of tables in a join path. If $t$ is in every $J_i$, we include it into $\bigcap_{J_i\in\Delta}S_i^\prime$.
We can now plug in this estimation in Theorem~\ref{Ajk} to compute the union size of joins in $\Delta$.  
Next, we compute the confidence interval for  $|\mathcal{O}_\Delta|$. The variance of $|\bigcap_{J_i\in\Delta}S_i^\prime|/|S_j^\prime|$, denoted by $\sigma_j^2$, can be computed by a binomial sampling, with a variance of $\hat{p_j}(1-\hat{p_j})$ and mean of $\hat{p_j}$~\cite{lee2008introduction}. Li et al. showed the mean and variance of $|J_j|$, denoted by $\phi_j^2$, are $T_{n}^j(u)(=\frac{1}{n-1}\sum_{i=1}^n f^j(i) )$ and  $T_{n,2}^j(u)(=\frac{1}{n-1}\sum_{i=1}^n (f^j(i)- T_n^j(f) )^2
)$, respectively~\cite{2016_wander_join}.  Assuming these terms are independent, we have the variance of $|\mathcal{O}_\Delta|$ as follows.  
$$\sigma_{|\mathcal{O}_\Delta|}^2=T_{n,2}^j(u) \cdot \hat{p_j} \cdot (1-\hat{p_j})+ T_{n,2}^j(u) \cdot \hat{p_j}+ T_{n}^j(u) \cdot \hat{p_j} \cdot (1-\hat{p_j})$$ 

This gives us the following confidence interval for  $|\mathcal{O}_\Delta|$ of Eq.~\ref{eq:onlineoverlap}. 
\begin{equation}\label{eq:confidence}
E = z\cdot\sqrt{
\begin{split}\frac{1}{n}\sum_{J_j\in\Delta} (&T_{n,2}^j(u) \cdot \hat{p_j} \cdot (1-\hat{p_j})+\\
&T_{n,2}^j(u) \cdot \hat{p_j}+ (T_{n}^j(u) \cdot \hat{p_j} \cdot (1-\hat{p_j})) 
\end{split}
}
\end{equation}

This means to obtain a $90\%$ confidence on overlap estimation, the algorithm requires a sample size of  $(\frac{1.96\cdot z}{E}\cdot\sigma_{|\mathcal{O}_\Delta|})^2$, on average.

Note that our estimator for overlap, using random walks, is unbiased. We first guarantee uniformity by adding $\frac{1}{p(t)}$ number of tuple $t$ to the collection $S_j$. We know we have the following. 
\begin{equation*}
    \lim_{|S_j| \to \infty}\frac{|\bigcap_{J_j\in\Delta}S_j^\prime|}{|S_j^\prime|} = \frac{\lim_{|S_j| \to |\Delta|}|\bigcap_{J_j\in\Delta}S_j^\prime|}{\lim_{|S_j| \to |\Delta|}|S_j^\prime|}= \frac{|\bigcap_{J_j\in\Delta}J_j|}{|J_j|}
\end{equation*}
Therefore, we can show that our result gets more and more accurate when $|S_j|$ gets larger and equals the exact result when $|S_j| = |\Delta|$. As the accuracy of overlap estimation gets closer to the true values, we also obtain a better estimation for the union size, which shows that our estimator improves for values used in our algorithms.

\eat{
\subsection{Discussion:  Direct vs. Random-walk Estimation}

We proposed a way of offsetting the cost of online estimation by reusing samples taken during the warm-up phase. The online estimation algorithm improves the  efficiency of sampling over the union of joins, by finding a tighter bound on join size and union size and having a lower rejection rate. However, it heavily relies on accessing the whole data and building a join graph which requires having index structures built in advance. Therefore,  online estimation is more applicable to settings of datasets residing in disk-based and in-memory database management systems, while direct estimation can be applied to data in the wild or scenarios, such as data markets, when access to the whole data is infeasible. }

\section{Online Union Sampling}
\label{sec:onlineunion}

The \direct method  has almost zero setup cost but  low sampling efficiency, while the \randomwalk method requires some sampling cost during the warm-up phase, but yields better estimation and  efficiency. 
To design a sampling algorithm with a minimal setup cost and high sampling efficiency, we introduce an online union sampling algorithm as illustrated in Algorithm~\ref{algsetback}. At a high level, join and union size estimation is performed in an online manner as the union of joins is being sampled. Algorithm~\ref{algsetback} extends Algorithm~\ref{algNB} with two optimizations: sample reuse and backtracking with parameter update. It initializes join and union parameters using the \direct method, then, continues with  selecting joins and sampling joins using the \randomwalk method. At each iteration, obtained samples are used to further refine estimations using the join and union estimation proposed in \S~\ref{sec:onlinejoinsize}.  

\textbf{Sample Reuse} (lines 8-10 of Algorithm~\ref{algsetback}) This makes up for the overhead of the \randomwalk.  Recall the  tuples sampled by \randomwalk  are not uniform, however, with an extra accept/reject step we can reuse them in the main sampling phase. For each join, we keep track of every tuple $t$ and its  probability $p(t)$, computed during join sampling  as described in \S~\ref{sec:onlinejoinsize}.  
Suppose we have already sampled $S = \{t_1, t_2, \dots, t_l\}, t_i \in J_j$, from $J_j$. Recall $S$ may have duplicates, i.e., there exists $i,j$ s.t. $t_i=t_j$. Then, if we choose $J_j$, we can first randomly choose a tuple $t$ from $t_1, t_2, \dots, t_l$, but we only accept it with probability $\frac{l}{p(t)\cdot|J_j|}$. In this way, the algorithm guarantees that the reused  $t$ is sampled from $J_j$ with probability $p(t)\cdot\frac{1}{l}\cdot\frac{l}{p(t)\cdot|J_j|} = \frac{1}{|J_j|}$ which ensures uniformity of sampling over the union. Note that if we accept $t$, we do not return $t$  to the pool, i.e., it is a sample without replacement process and $l$ is changing. Once we use all the tuples we stored, the next time $J_j$ is selected, we simply sample over join using the techniques of \S~\ref{sec:samplingjoin}.

Note that the acceptance rate, namely $R$,   can be equal to and greater than 
$1$. This means the algorithm may return more than one instance of $t$ in a certain round, while still ensuring the uniforming condition. 
We define the $r_i$ as the probability that $i$ instances of $t$ be accepted in a certain round. That is, $\sum_{i}^n r_i \cdot i =R$, where $\sum_{i}^n r_i = 1,~0\leq r_i \leq 1$. Then,  we have to choose the number of instances, $n \in N^+ $, by choosing one of the many valid solutions of this system. 

\eat{
\begin{algorithm}
\caption{Backtracking}
\begin{small}
\begin{algorithmic}[1]
\Procedure{Backtrack}{$T,\;P$}
\State $\{|J_j|', 1\ \leq j \leq n\}, |U|', \yl{conf\_level} \leftarrow $ {\sc OnlineUpdate} ($P$)
\For{each tuple $t$ in $T$}
\State remove $t$ from $T$ with probability $1 - \frac{|J(t.val)'|/|U|'}{|J(t.val)|/|U|}$
\EndFor
\State \textbf{return} $\{|J_j|', 1\ \leq j \leq m\}, |U|', T, \yl{conf\_level}$
\EndProcedure
\Procedure{OnlineUpdate}{$P$}
\For{each $J_j$}
\State $|J_j| \leftarrow |J_j| = \sum_{p \in P[j]} p \cdot \frac{1}{m}$
\EndFor
\State $|U|, \yl{conf\_level} \leftarrow $ {\sc CalcU}($P$) \Comment{Follow \S~\ref{sec:setunionsize} and \S~\ref{sec:online}}
\State \textbf{return} $\{|J_j|, 1\ \leq j \leq m\}, |U|, \yl{conf\_level}$
\EndProcedure
\end{algorithmic}
\label{algback}
\end{small}
\end{algorithm}
}

\begin{algorithm}
\caption{Set Union Sampling with Reuse and Backtracking}
\begin{small}
\begin{algorithmic}[1]
\Require{Join paths $\{J_j, 1\ \leq j \leq m\}$, tuple count $N$, backtrack para $\phi$, \yl{target confidence level $\gamma$}}
\Ensure{Tuples $\{t_i, 1 \leq i \leq N\}$}
\State $\{|J_j|, 1\ \leq j \leq n\}, |U|,  \leftarrow warmup(S)$, $\{|J_j'|\} \leftarrow\fn{cover(S)}$
\State $\yl{conf\_level \leftarrow 0}$, $T \leftarrow \{\}$
\Comment{result sample}
\State $P \leftarrow [j][]$
\Comment{record probability of selected tuples from each join path}
\State $orig\_join_i \leftarrow \{\}$ \Comment{record of original join of seen tuples}
\While{$n < N$}
\State{select $J_j$ with probability $\frac{|J_j'|}{|U|}$}
\If{$S_j \neq \emptyset$}
\State  sample $t\in S_j$, accept with  $\frac{l}{p(t)\cdot|J_j|}$, remove $t$ from $S_j$
\EndIf
\If{$S_j = \emptyset$ or $t$ from $S_j$ is rejected}
\State{$t \leftarrow$ a random sample from $J_j$}
\EndIf
\If{$t \in orig\_join_i $ for any $i < j$} reject $t$
\Else
\If{$t \in orig\_join_i $ for any $i > j$} 
\Comment{\fn{revision}}
\State remove $t$ from $orig\_join_i$ \fn{ and add $t$ to $orig\_join_j$}
\State remove all $t$'s from $T$ and delete $P[i][t]$'s
\EndIf
\If{$t \notin orig\_join_i$} add $t$ to $orig\_join_j$
\EndIf
\State $T \leftarrow T \cup \{t\}$ and update $P[j][t]$
\EndIf
\If{$\sum_{j\in[m]}|P[j]| \; \% \; \phi == 0$ \yl{and $conf\_level < \gamma$}} 
\State $\{|J_j|, 1\ \leq j \leq m\}, |U|, T$
\State $conf\_level \leftarrow  {\sc Update} (T, P)$
\Comment{backtrack (\S~\ref{sec:onlineunion})}
\EndIf
\EndWhile
\State \textbf{return} $T$
\end{algorithmic}
\label{algsetback}
\end{small}
\end{algorithm}

\textbf{Backtracking with Parameter Update} In \S~\ref{sec:joinoverlap}, we show that despite the small overhead of the  \direct method and its usefulness, the \direct method  may not be an unbiased estimator of our sampling parameters. 
Moreover, in \S~\ref{sec:online}, we proved that the \randomwalk  is an unbiased estimator whose parameter estimations converge to the true values after infinitely many numbers of samples. 
Algorithm~\ref{algsetback} initializes the framework with the estimation of the \direct method and  refines the parameters by  applying \randomwalk.  
The caveat is that, with this refinement strategy, although at each round the probability of sampled tuples  is uniform and equal to $1/|U|$,  the uniformity of tuples sampled across rounds is not guaranteed since the estimation of $|U|$ changes from one round to another with more random walks. To mitigate this non-uniformity, we introduce a backtracking trick which is an accept/reject strategy for all already sampled tuples in previous rounds.  

Algorithm~\ref{algsetback} initializes $T$ to be the set of  result samples and  initializes a list $P$ to store, $p(t)$'s, the probabilities of tuples obtained from a join  either it is accepted, rejected, or when the random walk fails. We also specify a parameter $\phi$, which indicates how often we backtrack. During the sampling process, 
we record all $p(t)$'s regardless of  $t$ being a rejected or reused tuple or being the  result of a failed random walk (in this case, $p(t) = 0$). 
Every $\phi$ iterations, i.e., $\phi$ recorded $p(t)$'s, we update join, overlap, and union estimations following the \randomwalk method then perform backtracking following Algorithm~\ref{algback} to adjust the probability of previously sampled tuples based on the new estimation of $|U|$. 

During backtracking, we iterate over all previously sampled tuples in the result  and adjust their probabilities by rejecting tuple $t$ with probability  $\frac{|J(t.val)^\prime|/|U|^\prime}{|J(t.val)|/|U|}$, where $|J(t.val)|$ and $|U|$ are original values, and $|J(t.val)|^\prime$ and $|U|$ are updated values. It is not hard to see that the backtracking algorithm guarantees that each tuple in the result is sampled with $\frac{1}{|U|^\prime}$. We also keep track of the confidence level $\gamma$ of the estimated sizes and stop backtracking when the accuracy is beyond a predefined threshold. 
\fn{\section{Other Types of Joins}}
\label{sec:others}

In this section, we show how to generalize our sampling framework to acyclic joins. 
The join  subroutine of our algorithm relies on an existing algorithm. The work by Zhao et al. provides a way of random sampling over join of all types: chain,  acyclic, and acyclic~\cite{2018_sample_join_revisit}. The two discussed instantiations of our framework propose different ways of estimating join overlap size parameters. The \randomwalk method relies on samples obtained from joins for estimation and handles acyclic and cyclic joins in the subroutine of join sampling.  
For brevity, we do not repeat the algorithm of Zhao et al. and describe how we extend the \direct method to acyclic and cyclic joins.  

\subsection{Acyclic Joins}
\label{sec:acyclic}
  
We organize the relations in a join tree, where each node refers to a relation 
and each edge  denotes a join. Figure~\ref{fig:acyclicjointree} illustrates an example of a join tree. 
The basic idea in extending our sampling algorithm to acyclic 
joins is to transform all acyclic joins and chain joins in the union to the base case of equi-length chain joins and use the results of \S~\ref{sec:overlapmulti} to estimate join overlaps. 
Our solution involves first building a {\em standard template} of joins. 
A template is a join tree structure  to which the structure of every  join can be converted. 
We formalize the standard template as a chain join that contains  relations of two attributes. 
The reason we need the template is that the degree-based comparison which is necessary for the size estimation of \S~\ref{sec:overlapmulti} can only be applied when relations have exactly the sample structure.

\newcommand{\rulesep}{\unskip\ \vrule\ }
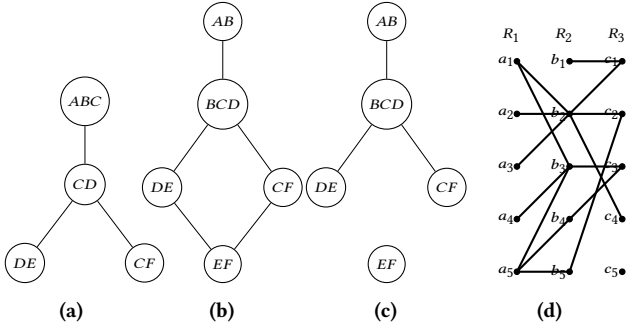
\begin{figure}
    \tiny
    \raggedright
    \begin{subfigure}[b]{.1\textwidth}
    \begin{tikzcd}[column sep=.5cm, row sep=.5cm, cells={nodes={draw, circle}}]
        &[-.3cm]A B C\arrow[dash]{d}&[-.3cm]\\
        &[-.3cm] C D\arrow[dash]{dl}\arrow[dash]{dr}&[-.3cm] \\
        D E && C F
    \end{tikzcd}
    \caption{}
    \label{fig:templatejoin}
    \end{subfigure}
    \begin{subfigure}[b]{.12\textwidth}
    \begin{tikzcd}[column sep=.5cm, row sep=.5cm, cells={nodes={draw, circle}}]
        &[-.3cm]A B\arrow[dash]{d}&[-.3cm]\\
        &[-.3cm]B C D\arrow[dash]{dl}\arrow[dash]{dr}&[-.3cm] \\
        DE\arrow[to=X, dash] && CF\arrow[to=X, dash]\\
        &[-.3cm]|[alias=X]| E F
    \end{tikzcd}
    \caption{}
    \label{fig:cyclicjointree}
    \end{subfigure}
    \begin{subfigure}[b]{.12\textwidth}
    \begin{tikzcd}[column sep=.5cm, row sep=.5cm, cells={nodes={draw, circle}}]
        &[-.3cm]A B\arrow[dash]{d}&[-.3cm]\\
        &[-.3cm]B C D\arrow[dash]{dl}\arrow[dash]{dr}&[-.3cm] \\
        DE && CF\\
        &[-.3cm]|[alias=X]| E F
    \end{tikzcd}
    \caption{}
    \label{fig:acyclicjointree}
    \end{subfigure}
\begin{subfigure}[b]{.12\textwidth}
\raggedleft
\begin{tikzpicture}[scale=0.7]
\draw[fill=black] (0,4) circle (1.5pt); 
\draw[fill=black] (1,4) circle (1.5pt); 
\draw[fill=black] (2,4) circle (1.5pt); 
\draw[fill=black] (0,3) circle (1.5pt); 
\draw[fill=black] (1,3) circle (1.5pt); 
\draw[fill=black] (2,3) circle (1.5pt); 
\draw[fill=black] (0,2) circle (1.5pt); 
\draw[fill=black] (1,2) circle (1.5pt); 
\draw[fill=black] (2,2) circle (1.5pt); 
\draw[fill=black] (0,1) circle (1.5pt); 
\draw[fill=black] (1,1) circle (1.5pt); 
\draw[fill=black] (2,1) circle (1.5pt); 
\draw[fill=black] (0,0) circle (1.5pt); 
\draw[fill=black] (1,0) circle (1.5pt); 
\draw[fill=black] (2,0) circle (1.5pt); 
\node at (-0.1,4.5) {$R_1$};
\node at (0.9,4.5) {$R_2$};
\node at (1.9,4.5) {$R_3$};
\node at (-0.2,4) {$a_{1}$};
\node at (0.8,4) {$b_{1}$};
\node at (1.8,4) {$c_{1}$};
\node at (-0.2,3) {$a_{2}$};
\node at (0.8,3) {$b_{2}$};
\node at (1.8,3) {$c_{2}$};
\node at (-0.2,2) {$a_{3}$};
\node at (0.8,2) {$b_{3}$};
\node at (1.8,2) {$c_{3}$};
\node at (-0.2,1) {$a_{4}$};
\node at (0.8,1) {$b_{4}$};
\node at (1.8,1) {$c_{4}$};
\node at (-0.2,0) {$a_{5}$};
\node at (0.8,0) {$b_{5}$};
\node at (1.8,0) {$c_{5}$};
\draw[thick]  (0,4) -- (1,3); 
\draw[thick]  (0,4) -- (1,2); 
\draw[thick]  (0,3) -- (1,3); 
\draw[thick]  (0,2) -- (1,3); 
\draw[thick]  (0,1) -- (1,2); 
\draw[thick]  (0,0) -- (1,2); 
\draw[thick]  (0,0) -- (1,1); 
\draw[thick]  (0,0) -- (1,0); 
\draw[thick]  (1,4) -- (2,4); 
\draw[thick]  (1,3) -- (2,4); 
\draw[thick]  (1,3) -- (2,1); 
\draw[thick]  (1,3) -- (2,3); 
\draw[thick]  (1,2) -- (2,2); 
\draw[thick]  (1,1) -- (2,2); 
\draw[thick]  (1,0) -- (2,3); 
\end{tikzpicture}
\caption{}
\label{fig:indexgraph}
\end{subfigure}
    \caption{(a) Acyclic join, (b) Cyclic join, (c) Tree structure for overlap estimation of (b), (d) Skeleton join  and residual joins for random sampling (d') Join data graph of $J$}
    \label{fig:mix}
\vspace{-5mm}
\end{figure}

To rewrite an acyclic join as a base chain join, we first construct the equivalent join tree such that a breadth-first traversal, always starting from the left-most node in each level, gives us  joins of the same  schema for all trees. 
A chain join is indeed a join tree with one branch. 
Joins may result in different tree structures. Therefore, we next need to choose a standard tree structure (template) before decomposing them  into base chain joins. A good template  is important in the estimation process. A bad template  can lead us to the worst bound results of $\min_{j\in[n]} |J_j|$. 

\begin{example} Consider the join in Fig.~\ref{fig:templatejoin}.  Suppose we choose the  template of $(A,D) \bowtie (A,C) \bowtie (B,C) \bowtie (B,E) \bowtie (E,F)$. To obtain $(B,E)$, we need to estimate the size of $(A,B,C) \bowtie (C,D) \bowtie (D, E)$; to obtain $(E, F)$, we need to estimate the size of $(D, E) \bowtie (C, D) \bowtie (C, F)$. Since we also need to estimate the fake join size, these two estimations between relations lose lots of information. However, the template $(A,B) \bowtie (B,C) \bowtie (C,D) \bowtie (D,E) \bowtie (E,F)$ gives us a better bound as we only use the pre-estimation for relations once to obtain $(E, F)$.
\end{example} 

It is not hard to notice that if we want to preserve most of the structure of the original relations, we prefer  templates that put  attributes in their original relations.  
We formulate the problem of finding a standard  template for a collection of chain and cyclic joins as the problem of splitting joins into  two-attribute relations such that the total pairwise distance of attributes in the same relation, in the tree of the template, is minimized. 

\subsubsection{Pairwise attributes score} 
Suppose all $J$'s in $S$ result in tables with attributes $\mathcal{D}$. For any pair of attributes $A, A^\prime \in \mathcal{D}$, let $Dist_j(A, A^\prime)$ be the distance between node(relation)s of $A$ and $A^\prime$ in join tree for $J_j$. Note that the distance between two attributes $A$ and $A^\prime$ is equivalent to the number of joins we need to perform   to obtain $(A, A^\prime)$ in a template. Then, we define the score between $A$ and $A^\prime$ as 
    $score(A, A^\prime) = \sum_{j\in[n]}Dist_j(A, A^\prime)$. 

Again consider Figure~\ref{fig:templatejoin}. We have $score(A, B) = 0 + 0 + 0 = 0$, which has the highest priority when we select a table for the standard. Moreover, $score(A, F) = 2 + 3 + 2 = 7$ represents that $A$ and $F$ are far from each other and have a small possibility to appear together in the original tables. Thus, pairs with a lower score have a higher possibility of originally being in the same table. 
The lower the score is, the higher the priority. We form all the pairs as a tree, where the root is an empty node and each path from the root to a leaf is an eligible path after eliminating the empty root node. For example, if the resulting table has schema $\mathcal{D} = \{A, B, C\}$, and $(A,B) = 0, (A,C) = 3, (B,C) = 6$, the tree will be formed as shown in Fig.~\ref{fig:splittree}. 
\begin{center}
\begin{tikzpicture}
[
    level 1/.style = {sibling distance = 1.8cm, level distance = 1cm},
    level 2/.style = {sibling distance = 0.9cm, level distance = 1cm
    }
]
\small
\node {Null}
    child {node {A B (0)}
    child {node {A C (3)}
    child {node {B C (6)}}}
    child {node {B C (6)}
    child {node {A C (3)}}}}
    child {node {A C (3)}
    child {node {A B (0)}
    child {node {B C (6)}}}
    child {node {B C (6)}
    child {node {A B (0)}}}}
    child {node {B C (6)}
    child {node {A B (0)}
    child {node {B C (6)}}}
    child {node {A C (3)}
    child {node {A B (0)}}}};
\label{fig:splittree}
\end{tikzpicture}
\end{center}
We want the standard template to have the lowest score, so we can convert the problem to finding the minimum cost path which can be solved recursively. 

\subsubsection{Alternating score} Another thing worth noticing is that split relations and joins without estimating sub-join size preserve most information, so we may give weights to the case with $Dist_j(A, A^\prime) = 0$. We can view the score for this case as a hyper-parameter that can be tuned for finding the tightest bound.

Given a standard template, we now introduce how acyclic and cyclic joins can be converted while preserving information for "fake join"s. Consider the tree structure acyclic join. Suppose node for $R_i$ has $k$ number of children, $R_{i_1}, R_{i_2}, \dots, R_{i_k}$, and we have an extreme case of the template where each table $R_{i_j}$ has one attribute that is paired with an attribute in $R_i$. In this case, we do {\em fake join} on each $R_{i,j}^\prime = R_i \bowtie^\prime R_{i_j} \bowtie^\prime Childs(R_{i_j})$ and estimate $|R_{i,j}^\prime|$ using the method in \S~\ref{sec:joinsize}. In this step, we also record the estimated maximum degree in each attribute $A$ in $R_{i_j}^\prime$ as follows:
\begin{align*}
    M_{A}(R_{i_j}^\prime) = \left\{\begin{array}{l}
    M_{A}(R_i) \cdot M_{A}(R_{i_j})\; \text{ if }\; A \text{ is join attribute}\\
    \\
    \max\{M_{A}(R_i), M_{A}(R_{i_j})\}\; \text{ otherwise }
\end{array}\right.
\end{align*}
Through this way, we can split $R_{i_j}^\prime$ according to the standard template and with information on both cardinality and maximum degrees. Moreover, we are able to estimate the overlap size accordingly. Note that we do not necessarily need to fake join all the child nodes with their parent for transformation, as in real scenarios, we select the children based on the schemes of relations in the standard template.

\subsection{Cyclic Joins}
\label{sec:cyclic}

In this section, we extend our sampling algorithm to cyclic queries.
Following the method proposed in~\cite{2018_sample_join_revisit}, we break  all the cycles in the join hyper-graph by removing a 
subset of relations so that the join becomes a connected and acyclic
join. 
The residual join, namely $\mathcal{S}_R$, is the set of removed relations and the skeleton join, namely $\mathcal{S}_M$, is the set of relations in the main acyclic join. Fig.~\ref{fig:acyclicjointree} shows the equivalent skeleton join tree and residual join to the cyclic join of Fig.~\ref{fig:cyclicjointree}. 
Let attributes in $\mathcal{S}_R$ be $Attr(\mathcal{S}_R)$, and attributes in $\mathcal{S}_M$ as $Attr(\mathcal{S}_M)$. We treat $\mathcal{S}_R$ as a single relation in the new acyclic join. We can even materialize $\mathcal{S}_R$ by performing joins in $\mathcal{S}_R$. Note that some attributes in $Attr(\mathcal{S}_R)$ from the residual $\mathcal{S}_R$ may be joined with $Attr(\mathcal{S}_M)$. This means we have an acyclic join (the skeleton join) and a residual that can be joined with two or more relations in the skeleton. Now the maximum degree $M(\mathcal{S}_R)$ of any attribute in $\mathcal{S}_R$  is defined as follows. 
\begin{equation*}
    M(\mathcal{S}_R) = \max_{v_i \in A_i}\vert t:t\in \mathcal{S}_R, \pi_{A_i}(t) = v_i, \forall A_i \in Attr(\mathcal{S}_M) \cap Attr(\mathcal{S}_R)\vert
\end{equation*}
Since we treat the residual as one relation, 
with the degree information, we can estimate the join size and  overlap size by breaking  $S_R$ into the base chain join structure, as described in \S~\ref{sec:acyclic}. 
Note that  the choice of set or relations to remove can have a significant influence on performance.  We follow the methods used by Zhao et al.~\cite{2018_sample_join_revisit} to decide where to break the cycle  in practice.

\subsection{Selection Predicates}
\label{sec:filterpred}

Our sampling algorithms can  support selection predicates in  two ways. The first alternative is to push down the predicates to relations, i.e., we filter each relation with the predicates, during the pre-processing, and work with filtered relations during sampling. 
This paradigm works for both \direct and \randomwalk.  
Another alternative is to enforce the selection predicate during the sampling process. This paradigm works with only \randomwalk, unless the  \direct method has access to the selectivity degree of the predicate and can adjust the degree statistics. Since this paradigm adds an additional rejection factor, it is most appropriate for selection predicates that are not very selective. 
\section{Evaluation}
\label{sec:eval}

\begin{figure*}[!ht]
    \begin{subfigure}[t]{0.245\linewidth}
        	\centering
            \includegraphics[width=\linewidth]{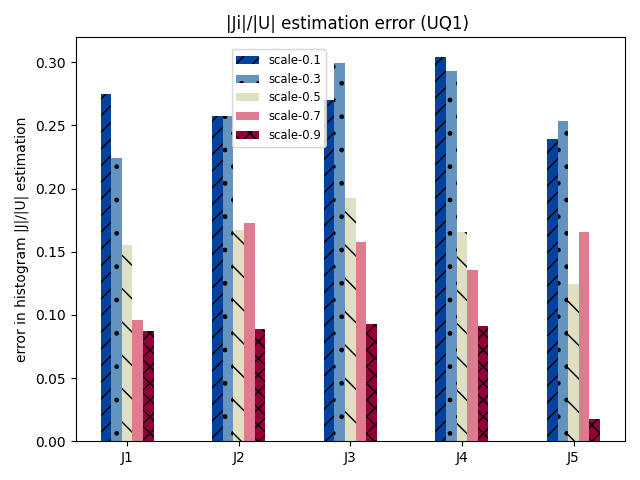}
        	\vspace{-6mm}
        	\caption{}
            \label{fig:errorration-uq1}
    \end{subfigure}
    \begin{subfigure}[t]{0.245\linewidth}
        	\centering
            \includegraphics[width=\linewidth]{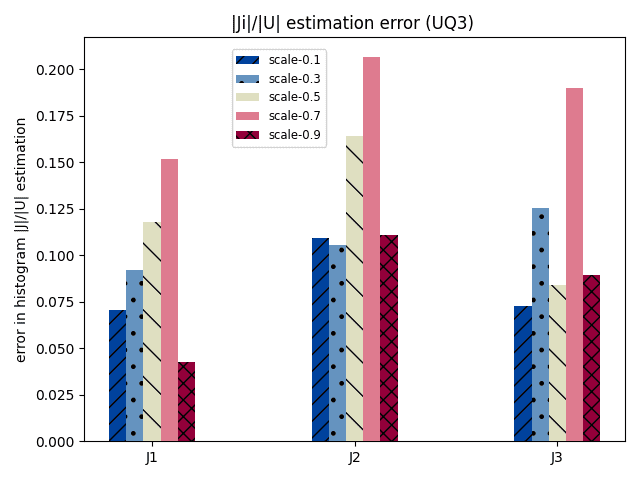}
        	\vspace{-6mm}
        	\caption{}
            \label{fig:errorration-uq3}
    \end{subfigure}
    \begin{subfigure}[t]{0.245\linewidth}
        	\centering
            \includegraphics[width=\linewidth]{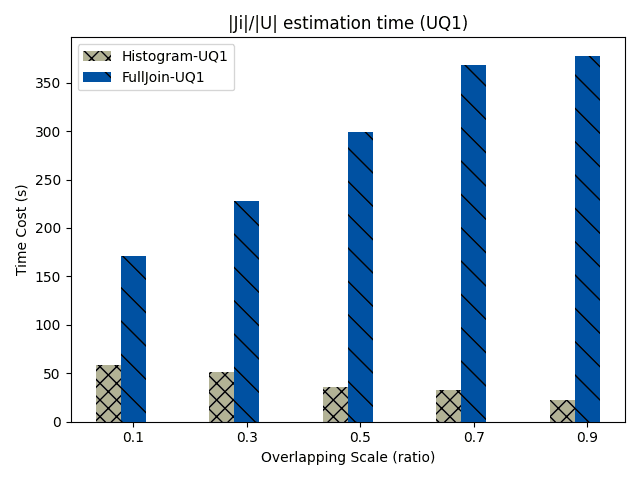}
        	\vspace{-6mm}
        	\caption{}
            \label{fig:union-time-uq1}
    \end{subfigure}
    \begin{subfigure}[t]{0.245\linewidth}
        	\centering
            \includegraphics[width=\linewidth]{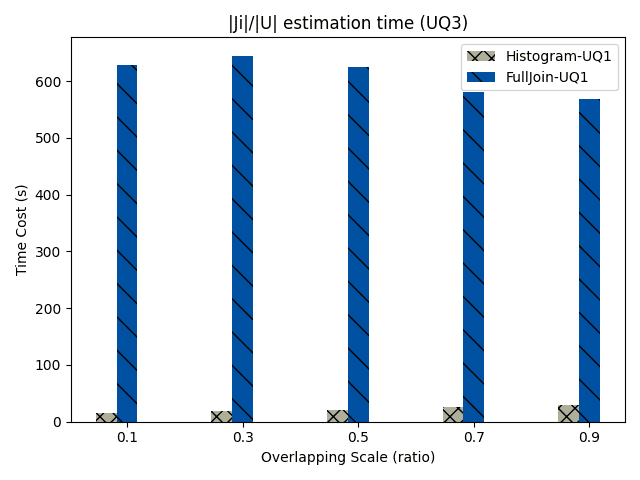}
        	\vspace{-6mm}
        	\caption{}
            \label{fig:union-time-uq3}
    \end{subfigure}
\caption{The error of join to union size ratio estimation using \direct+{\sc EO} on (a) UQ1 and (b) UQ3; runtime of union size estimation using \direct and {\sc FullJoin} on (c) UQ1 and (d) UQ3.} 
    \label{fig:errorration-a}
\end{figure*}

\textbf{Datasets:} We use three datasets consisting of different types of joins tailored from the TPC-H benchmark. Each query workload is to sample from the union of joins in a dataset. \yl{{\sc UQ1} consists of five chain joins, where each has five relations: {\tt \small nation}, {\tt \small supplier}, {\tt \small customer}, {\tt \small orders}, and {\tt \small  lineitem}; {\sc UQ2} consists of three chain joins which use: {\tt \small region}, {\tt \small nation}, {\tt \small supplier}, {\tt \small partsupp}, and {\tt \small part}, where we also add selection predicates following $Q_2^N \cup Q_2^P \cup Q_2^S$ in~\cite{CarmeliZBKS20}; and, {\sc UQ3} has one acyclic join and two chain joins.} {\sc UQ3} is derived from  relations:  {\tt \small supplier}, {\tt \small customer}, and {\tt \small orders}. We split them vertically and horizontally to get  relations with different schemas. Therefore, working with {\sc UQ3} involves the application of the splitting method. 

To experiment with the scale of data, we use TPCH-DBGen to generate relations with various scales. For example, with TPC-H scale factor $N$-gb, and $K\%$ scale ratio, {\sc UQ3} is a dataset of size $K\% \cdot N \cdot 3$. \yl{For {\sc UQ2}, we have the same data for three joins but have different constraints for selection predicates. Hence, {\sc UQ2} has a large overlap scale.} We also vary the overlap scale $P\%$ between joins of  {\sc UQ1}. When generating different queries, we keep $P\%$ of the data the same in the original corresponding relations. This way, although we cannot ensure that the overlap ratio in queries is exactly $P\%$,  given unknown information between relations, we can guarantee that the overlap ratio between queries is  proportional to the overlap scale. Note that we did not perform experiments on  cyclic joins queries, particularly because transforming cyclic  to acyclic joins and online sampling from cyclic join is done based on an existing work~\cite{2018_sample_join_revisit}. 

\noindent\textbf{Algorithms} \yl{
We evaluate the \direct and a \randomwalk instantiations 
by plugging in techniques 
of \S~\ref{sec:joinoverlap} and 
\S~\ref{sec:online}, respectively, in Theorem~\ref{Ajk}. 
The join estimation of \direct can be instantiated by baselines {\sc EW} (Exact Weight)~\cite{2018_sample_join_revisit}, which is the ground truth for weights by calculating the exact weight of each tuple in the join data graph, or {\sc EO} (Extended Olken's)~\cite{2018_sample_join_revisit}, which we described in \S~\ref{sec:samplingjoin}.
The join estimation of the online technique uses our  \randomwalk of \S~\ref{sec:online}. 
We also consider {\sc FullJoinUnion} as the ground truth for our join size and union size estimations. This algorithm performs the full join and computes the union.}
Note that {\sc FullJoinUnion} is extremely expensive on large datasets. Our experiments timed out on data sizes of more than 5GB (per relation). 
We do not evaluate {\sc DisjoinUnion} since it is 
consistent with sampling over one join path as it has no extra delays. we do not evaluate the Bernoulli set union sampling since it is a slightly different variation of the Non-Bernoulli and has lower efficiency theoretically. 

\noindent\textbf{Implementation:}  The framework  is implemented in Python.  Relations in joins are stored in hash relations with a linear search. Acyclic joins are implemented in a tree structure and acyclic joins are handled by recursion. 
All experiments are conducted on a machine
with 2 Intel\textsuperscript{\textregistered} Xeon Gold 5218 @ 2.30GHz (64 cores), 512 GB DDR4 memory, a Samsung\textsuperscript{\textregistered} SSD 983 DCT M.2 (2 TB), 4 GPUs - TU102 (GeForce RTX 2080 Ti).

\subsection{Join and Union Size Approximation}

\subsubsection{Error} 
We evaluate the estimation error of the ratio $|J_i|/|U|$ for each join in a query, because our algorithms rely on this ratio to define probability distributions over joins. 
For these experiments, we use {\sc UQ1} and {\sc UQ3} with 3GB scale raw data. After preprocessing, {\sc UQ1} is $9$GB and {\sc UQ3} is $5.4$GB.The overlap scale is set to $0.2$. Fig.~\ref{fig:errorration-uq1} and~\ref{fig:errorration-uq3} show the ratio estimation error for {\sc UQ1} and {\sc UQ3}, with respect to overlap scale, using \direct method. 

{\em For large overlap scales, the error tends to be small and stable. For smaller scales, the performance is unstable.}
This is because when the overlap scale is small, 
small samples will have a large effect on the estimation performance. However, when we have a large scale of overlap, which is  our use case, 
the randomness will be removed. 
Besides, we observe that the average error for {\sc UQ3}, in  Fig.~\ref{fig:errorration-uq3}, is better than {\sc UQ1}, in  Fig.~\ref{fig:errorration-uq1}. As we take an upper bound for every join, our \direct  method gains higher accuracy on joins with a smaller length. Given that {\sc UQ3} is smaller both in length and numbers, this explains why the estimation is relatively more accurate for {\sc UQ3}. 

\eat{
\todo{next paragraph candidate for cutting} 

Additionally, we employ the split method on {\sc UQ3}. In practice, splitting relations does not affect the accuracy of estimating overlaps. However, for acyclic and cyclic joins, integrating relations, viewing them as a whole, and splitting them  incurs  some information  loss. In most cases, if we can preserve the main branch of the join and pre-process on smaller branches, and decrease the number of integration, the split method will have little impact on estimating the overlap. 
}

\subsubsection{Runtime} 
We  report the runtime of our parameter estimation  methods, in Fig.~\ref{fig:union-time-uq1} and~\ref{fig:union-time-uq3}. 
First, {\em \direct is significantly faster than the brute-force full join.}  
Second, for {\sc UQ1}, we observe that as the cost of full join increases with overlap scale, the time \direct method needs  becomes less. This is because when the overlap scale is large and the overlapping structure is complex, it becomes harder for the full join to scan over data, but for our method, a higher overlap scale instead accelerates our method in finding the tuple with the maximum degree. 

\yl{Unlike the \direct technique, our \randomwalk technique collects sampling statistics during the warm-up phase. When evaluating the confidence level of the overlap size, we are actually evaluating the ratio that the overlap part takes in the join, i.e.,  $\frac{|\bigcap_{J_i\in\Delta}S_i^\prime|}{|S_j^\prime|}$. In Eq.~\ref{eq:onlineoverlap}, we take $|J_j|$ as an exact value to fulfill the assumption of independence. This is in fact equivalent to having the confidence level of $|J_j|$ as $1$ and confidence interval as $0$, which is an approximation of the case given by Wande Join~\cite{2016_wander_join}. We terminate online sampling when the confidence level reaches $90\%$ or we obtain 1,000 samples.}

\yl{Fig.~\ref{fig:uq1_overlap} compares the performance of \direct with {\sc EO}~\cite{2018_sample_join_revisit, Olken}, as join size instantiation, with \randomwalk, in terms of the error of join to union size ratio estimation on {\sc UQ1}. We used a data scale of 3GB for each query. 
First, {\em \randomwalk outperforms \direct; in fact, \randomwalk is extremely accurate and stable and has an error close to zero for all joins.}  This is because the nature of indexing will give us extra information about overlapping.
We remark that the accurate estimation comes at the cost of sampling during the warm-up phase. We will discuss the empirical evaluation of the sampling technique that reuses these samples, shortly. Besides, {\em while the estimation error is quite robust across joins, the higher the overlap, the more accurate  \direct becomes.} Since the accuracy of overlap size estimation heavily depends on the overlap size of samples we collect, the larger the actual overlap is, the easier we find overlap in samples, As we take the minimum in each step as the upper bound of overlap size, the bound gets tighter when overlap size approaches data size, which results in more accurate results in overlap size estimation. Nevertheless, though \randomwalk has better performance, \direct is relatively faster and can be applied to databases without index structures. }

\subsection{Set Union Sampling}

\begin{figure*}[!ht]
    \begin{subfigure}[t]{0.245\linewidth}
        	\centering
            \includegraphics[width=\linewidth]{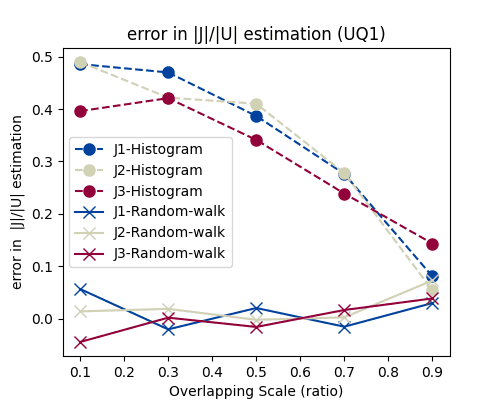}
        	\vspace{-6mm}
        	\caption{}
            \label{fig:uq1_overlap}
    \end{subfigure}
      \begin{subfigure}[t]{0.245\linewidth}
        	\centering
            \includegraphics[width=1.08\linewidth]{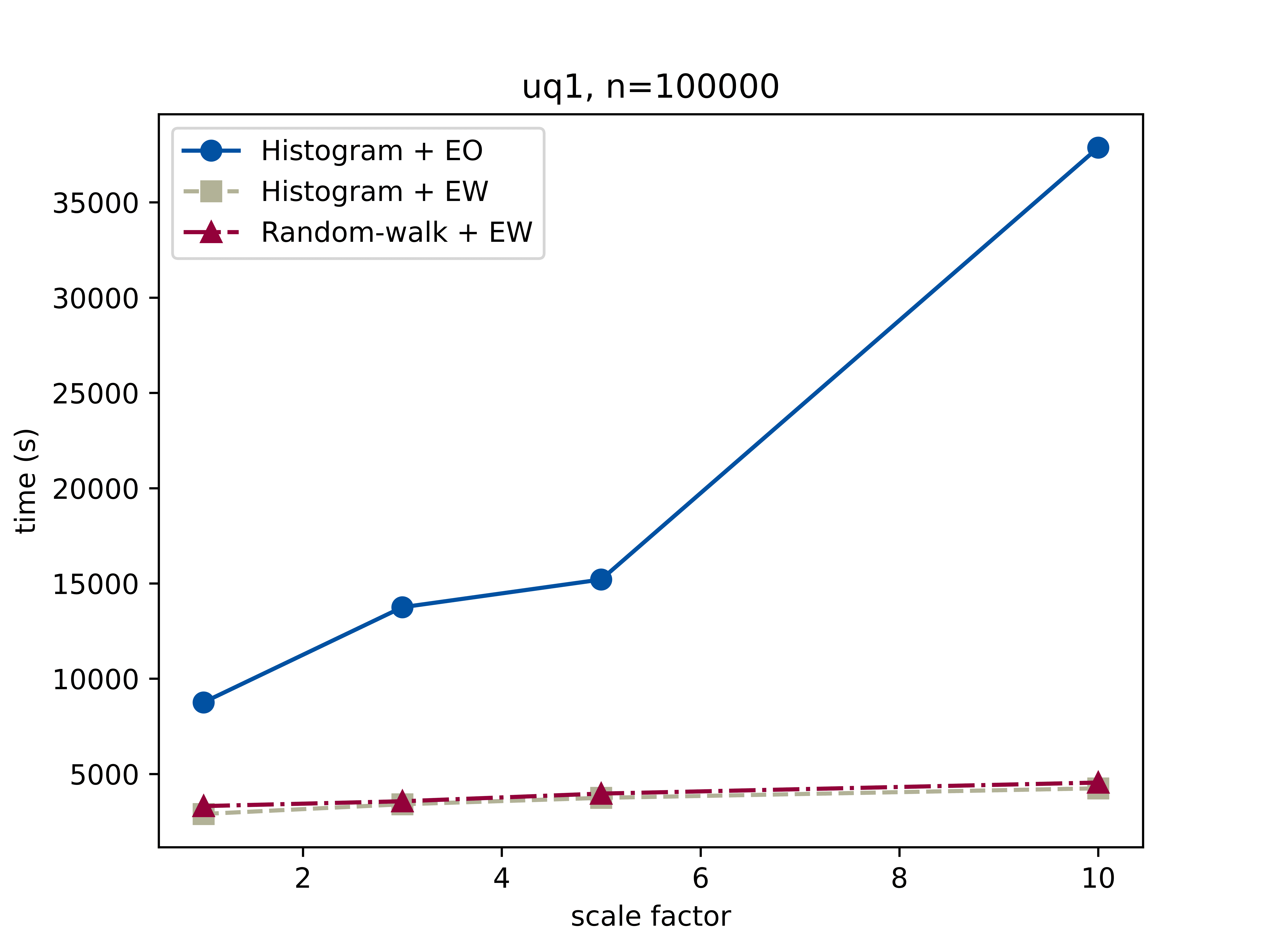}
        	\vspace{-6mm}
        	\caption{}
            \label{fig:uq1_nb_5_scale}
    \end{subfigure}
    \begin{subfigure}[t]{0.245\linewidth}
        	\centering
            \includegraphics[width=1.08\linewidth]{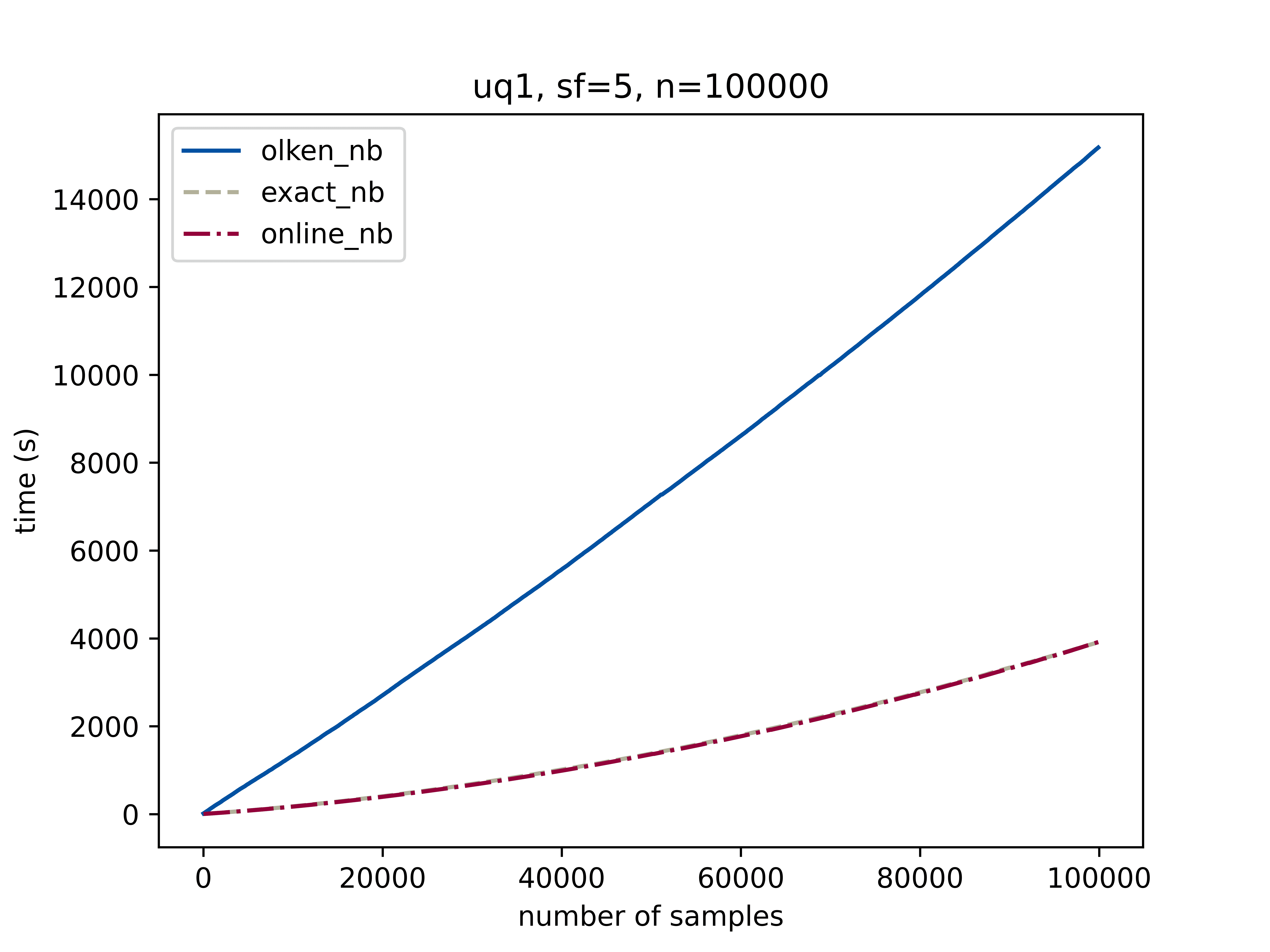}
        	\vspace{-6mm}
        	\caption{}
            \label{fig:uq1_nb_5_time}
    \end{subfigure}
    \begin{subfigure}[t]{0.245\linewidth}
        	\centering
            \includegraphics[width=1.08\linewidth]{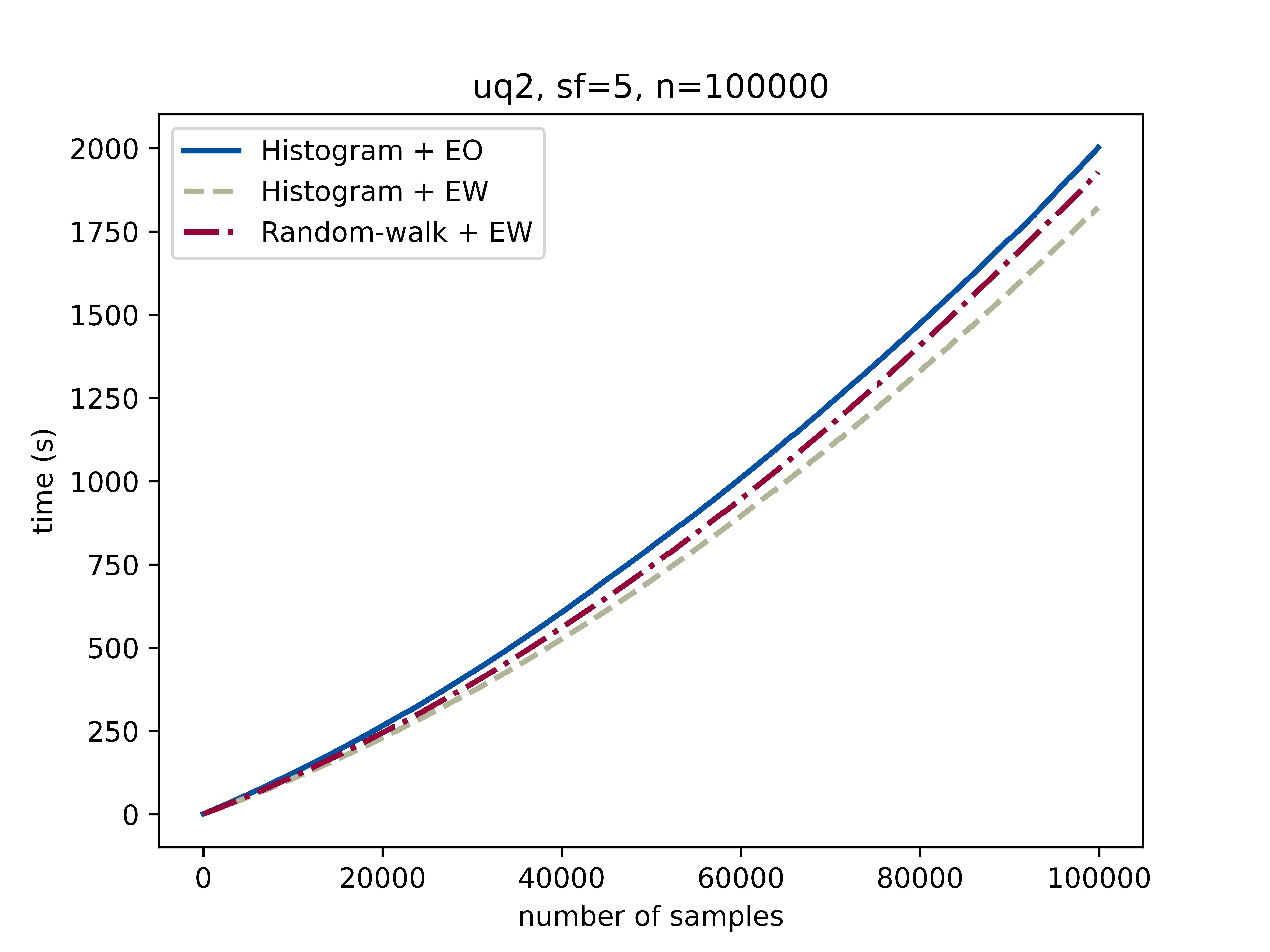}
        	\vspace{-6mm}
        	\caption{}
            \label{fig:uq2_nb_5_time}
    \end{subfigure}
    \hfill
    \begin{subfigure}[t]{0.245\linewidth}
        	\centering
            \includegraphics[width=1.08\linewidth]{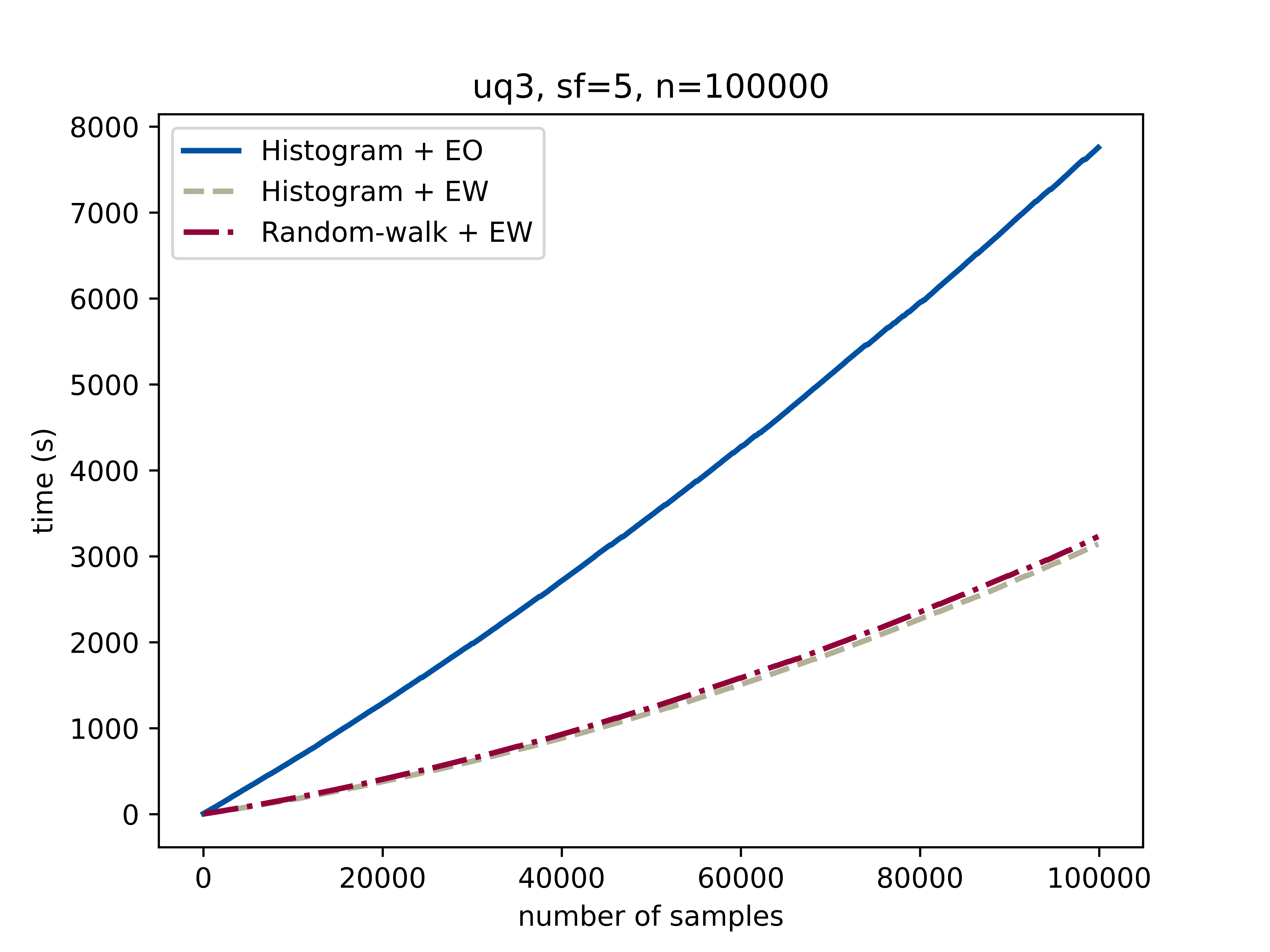}
        	\vspace{-6mm}
        	\caption{}
            \label{fig:uq3_nb_5_time}
    \end{subfigure}
    \begin{subfigure}[t]{0.245\linewidth}
        	\centering
            \includegraphics[width=\linewidth]{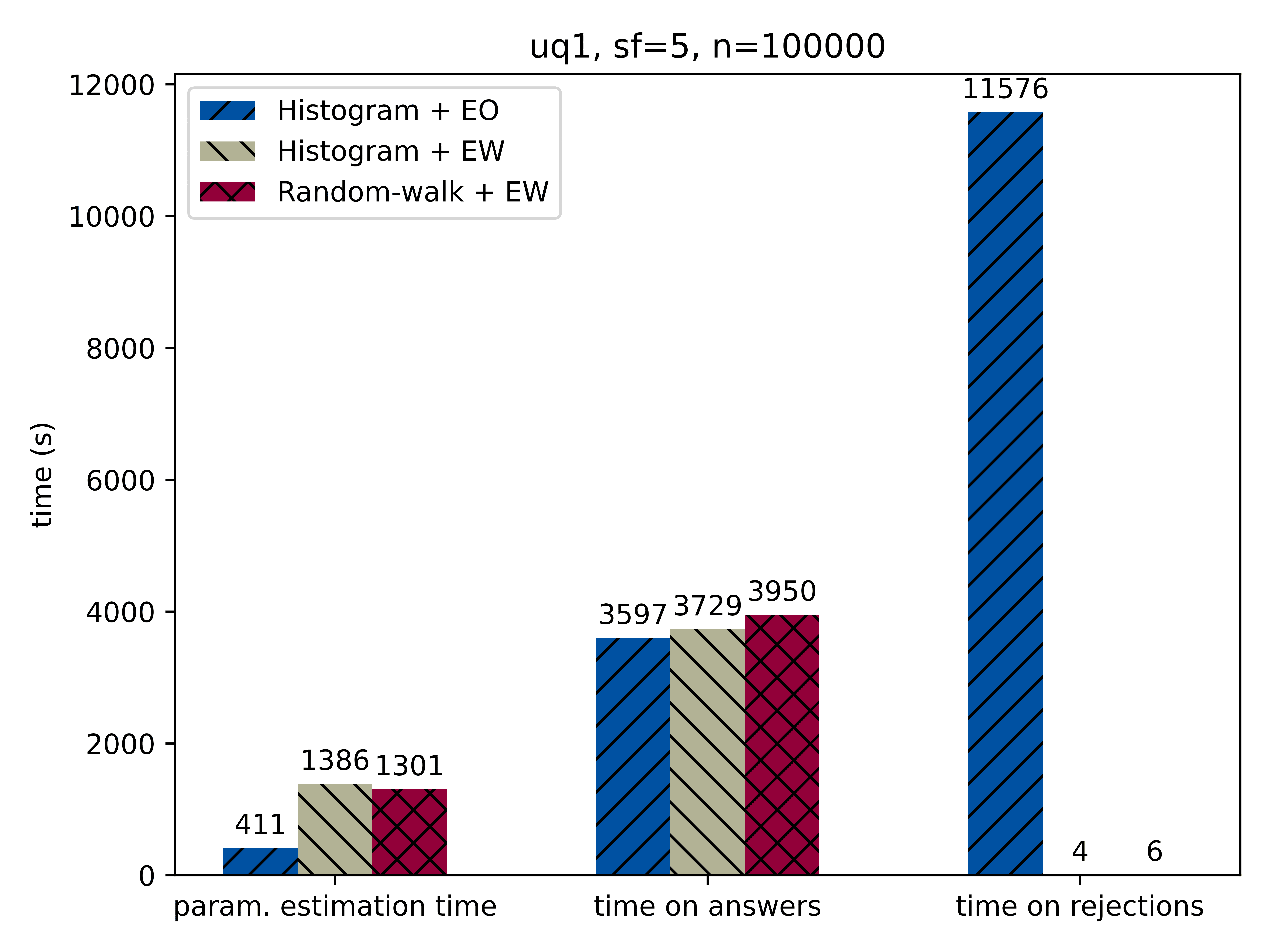}
        	\vspace{-6mm}
        	\caption{}
            \label{fig:uq1_nb_5_ratio}
    \end{subfigure}
      \begin{subfigure}[t]{0.245\linewidth}
        	\centering
            \includegraphics[width=\linewidth]{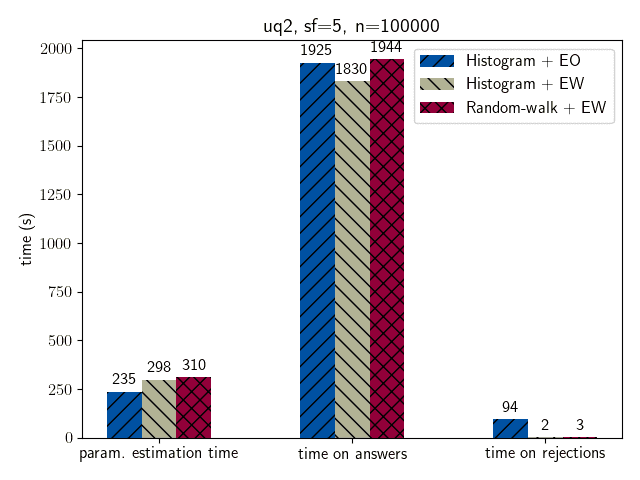}
        	\vspace{-6mm}
        	\caption{}
            \label{fig:uq2_nb_5_ratio}
    \end{subfigure}
    \begin{subfigure}[t]{0.245\linewidth}
        	\centering
            \includegraphics[width=\linewidth]{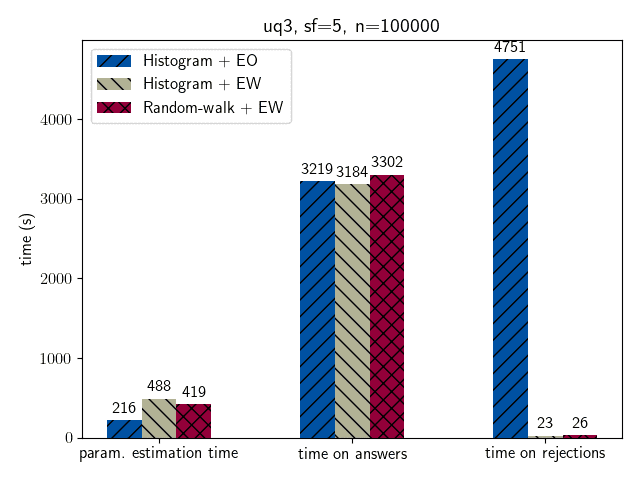}
        	\vspace{-6mm}
        	\caption{}
            \label{fig:uq3_nb_5_ratio}
    \end{subfigure}
    \caption{(a) the error of join to union size ratio; 
    (b) {\sc SetUnion} time vs. data scale on {\sc UQ1};
    runtime vs. sample size on (c)  {\sc UQ1} and (d) {\sc UQ2}; (e) {\sc UQ3}; time breakdown of (f) {\sc UQ1} (g) {\sc UQ2} and (h) {\sc UQ3}.}
    \label{fig:disunion-time}
\end{figure*}

\subsubsection{Scaling with Number of Samples} 

\sloppy{
For \direct, We use both {\sc EW} and {\sc EO} methods for weights initialization in sampling from a single join, and we only use {\sc EW} for \randomwalk. 
First, Fig.~\ref{fig:uq1_nb_5_time},~\ref{fig:uq2_nb_5_time}, and~\ref{fig:uq3_nb_5_time} show how {\sc SetUnion} scale with number of samples. 
Overall, we can see that when using {\sc EW} instantiation, \direct and \randomwalk have nearly no difference in performance. In other words, {\em the accuracy of the estimation bound has little impact on sampling efficiency.} However, for \direct, using {\sc EW} results in a much slower situation than using {\sc EO} on all three queries, since with exact weights calculated, we obtain a rejection rate of zero.} 

\subsubsection{Runtime Breakdown}

Fig.~\ref{fig:uq1_nb_5_ratio},~\ref{fig:uq2_nb_5_ratio}, and~\ref{fig:uq3_nb_5_ratio} shows the comparisons of time spent on parameter  estimation(join size, overlap, and weights), producing accepted answers and on producing rejected answers. The reason for the decay comes from {\sc EO}, as well as the fact that we need to reject duplicate tuples that are sampled from a join  different from what it is assigned to. From these  plots, the most significant finding is that though using {\sc EO} is much less efficient than using {\sc EW}, it has better performance in the warm-up phase. Moreover,   since it uses the upper bound of weights for sampling from a join, it has an extra rejection phase and needs to spend much more time on rejected answers than using {\sc EW}. Besides, the time spent on accepted answers is similar for three combinations of instantiations for all queries. Moreover, 
{\em our {\sc SetUnion} algorithm spends minor time rejecting duplicate tuples and has very high efficiency when using {\sc EW} for join sampling}. 

\subsubsection{Scaling with Relation Size}

\yl{Although we use $scale factor = 5$ for all three queries, we will get different sizes of unions if we perform full joins due to different numbers of relations and different levels of overlaps. From our set out, the order of union size for three queries from large to small is {\sc UQ1}, {\sc UQ3}, {\sc UQ2}. From both sets of plots, we notice that sampling time is in proportion to the resulting union size. What's more, when the expected union size is small, as for {\sc UQ2}, {\sc EO} has a relatively smaller gap with {\sc EW} during sampling, and has an even better advantage in the warm-up phase.}

\yl{Moreover, Fig.~\ref{fig:uq1_nb_5_scale} reports sampling time for various data scales for {\sc UQ1}. The first observation is that using {\sc EO} for join size estimation makes both algorithms slower than using {\sc EW} and overall {\sc EW} scales better with the size of data since with exact weights the rejection rate for sampling from single join path is $0$. Second, though the sampling time of both algorithms increases with the size of data, the scale has a much larger effect on {\sc EO} than {\sc EW}. As the size of each relation grows larger, a tuple has a higher rejection rate due to the growth of the number of tuples in the relation to be joined with. Finally, initialization in union size using either \direct or \randomwalk has little impact on efficiency, which is consistent with the conclusion we obtained earlier.}

\subsection{Online Union Sampling with Sample Reuse}
\label{sec:evalonline}

\begin{figure}[]
\begin{subfigure}[t]{0.495\columnwidth}
        	\centering
            \includegraphics[width=1.08\linewidth]{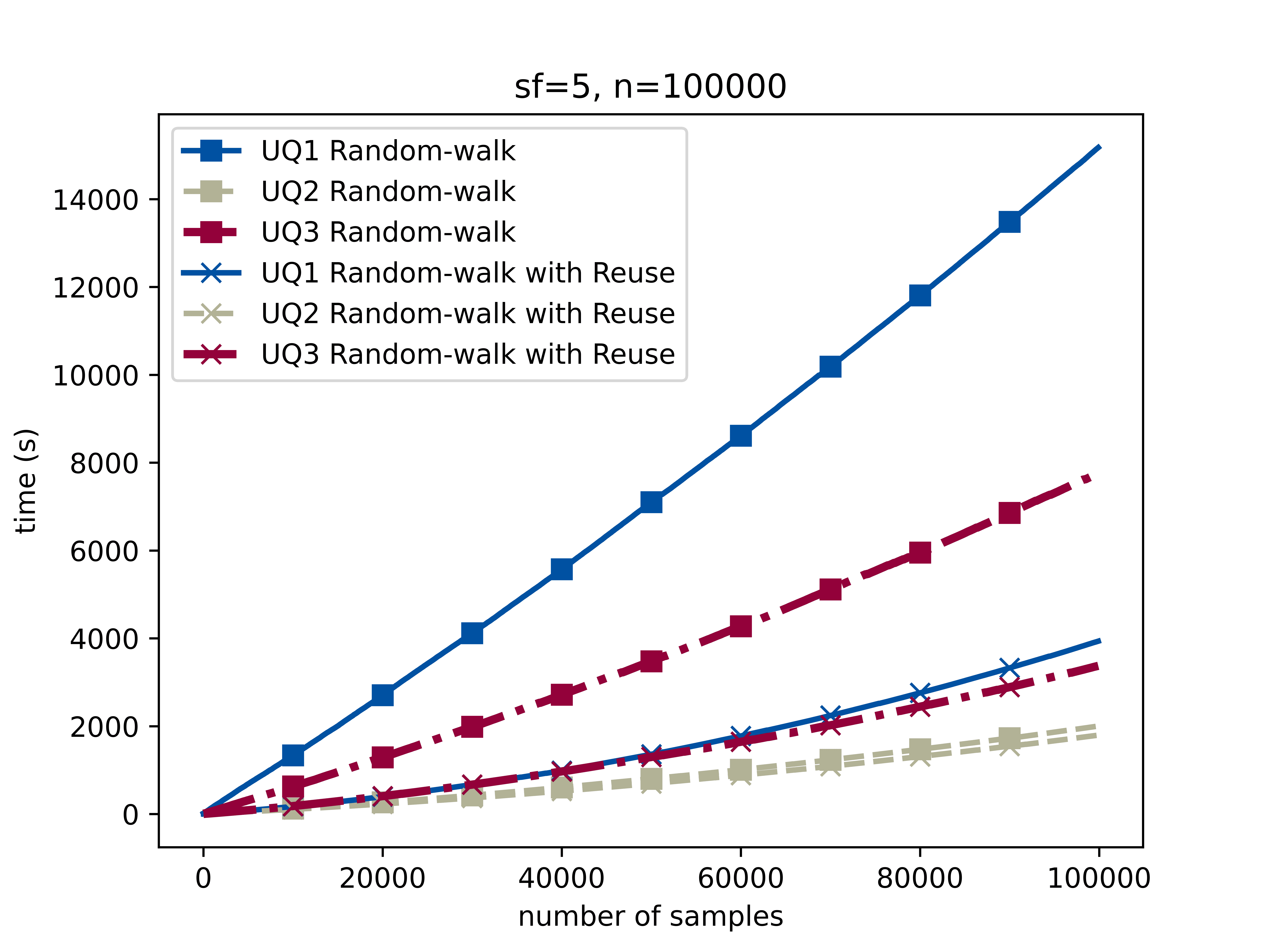}
        	\vspace{-6mm}
        	\caption{}
            \label{fig:reuse_time}
    \end{subfigure}
\begin{subfigure}[t]{0.495\columnwidth} 
\centering
\includegraphics[width=1.08\linewidth]{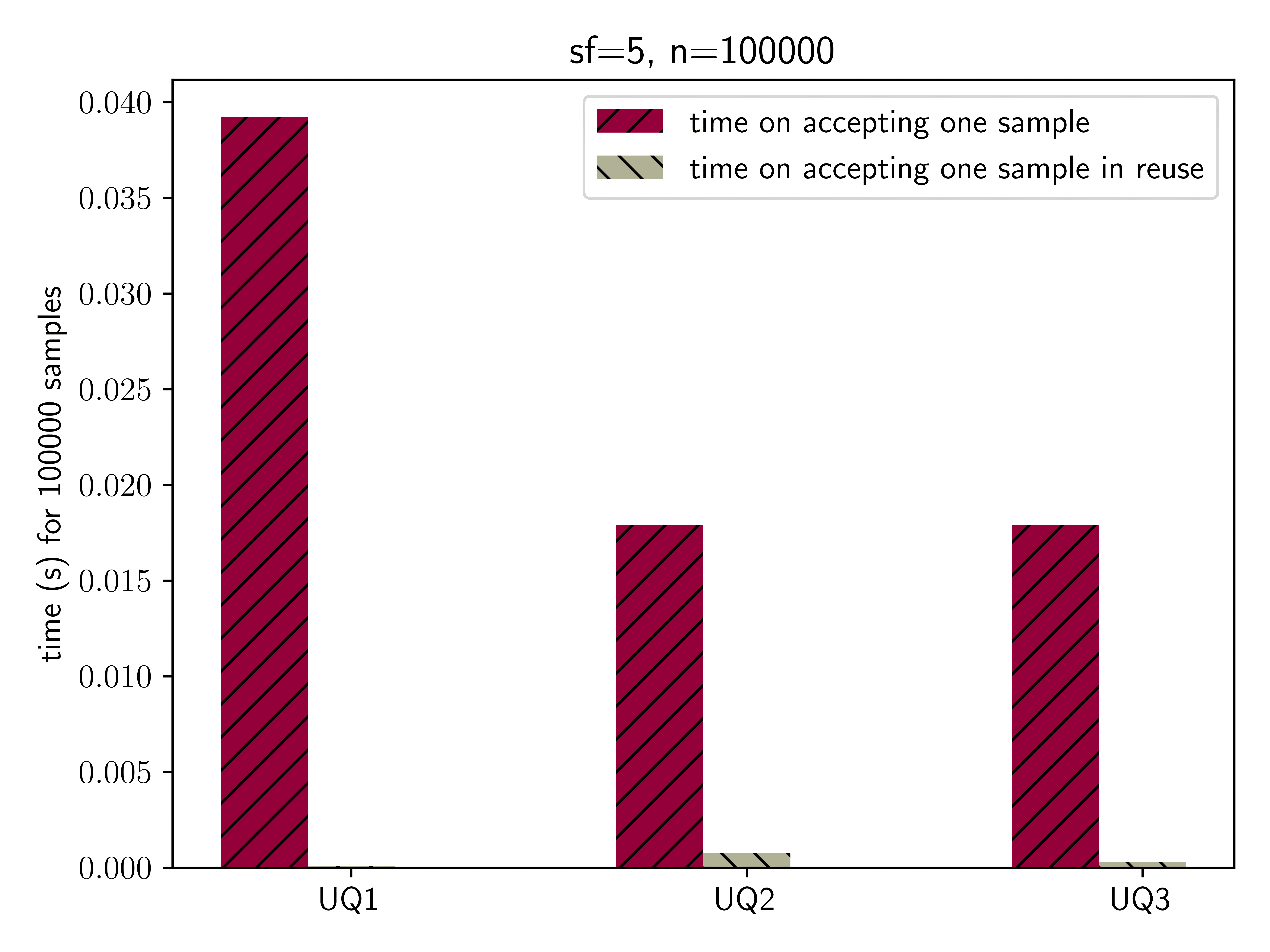}
\vspace{-6mm}
        	\caption{}
    \label{fig:reuse_ratio_per_sample}
\end{subfigure}
    \caption{(a) time vs. sample size 
    with and without reuse (b) time per sample spent in a regular phase vs. a reuse phase.}
    \label{fig:reuse}
    \vspace{-5mm}
\end{figure}

\yl{In the next set of experiments, 
we evaluate the runtime of the \randomwalk sampling using the idea of reusing samples collected during warm-up. We compare \randomwalk with and without reuse on all three queries. Fig.~\ref{fig:reuse_time} shows sampling time with respect to sample size. First, we can clearly observe that we have much higher efficiency when we sample with reuse. When we sample from the pool of pre-sampled and joined tuples during the warm-up phase, we only do a fast check on rejection or acceptance and do not need to sample over each relation. 
Moreover, there is a slight change in slope on lines of sampling with reuse cases. When pre-sampled samples are all used, 
the performance of {\sc SetUnion} will slowly converge to their original performance. One other interesting phenomenon is that the reuse of samples has a more apparent increase in performance when the expected union size is larger. For {\sc UQ1}, there is a huge gap between with and without reuse; but for {\sc UQ2}, the gap is much smaller. 
Fig.~\ref{fig:reuse_ratio_per_sample} compares to time spent on successfully accepting one tuple in the regular sampling phase and in the reuse sampling phase. We use the ratio of total time spent on sampling and the number of successfully sampled tuples for each phase for comparison, and we can see that when we reuse pre-sampled tuples, we have much higher efficiency. This shows the huge improvement in efficiency brought by our online union sampling.  
}

\section{Related Work} 
\label{sec:relatedwork}

\textbf{Random Access to Query Results} The closest problem to ours is random access to the results of conjunctive queries. 
Bagan et al. show that the free-connex acyclic conjunctive queries can be evaluated using an  enumeration algorithm with a constant delay between consecutive answers, at the cost of a linear-time preprocessing phase~\cite{BaganDG07}. However, because this work does not guarantee the randomness of the intermediate answers, the produced result may have extreme bias, making it unsuitable for learning tasks. Recently, Carmeli et al. studied the problem of enumerating the answers of the union of acyclic conjunctive queries in a uniformly random order~\cite{CarmeliZBKS20}. The proposed algorithm requires full access to the database, i.e., the computation of the full joins as well as a linear pre-processing time in the size of the database. As such, this algorithm is not applicable to random sampling over open data, data markets, proprietary databases, or web databases where the access model is tuple-at-a-time access. Unlike the approach of  Carmeli et al. which requires  computing the exact join  and overlap sizes, our framework presents  sampling strategies and ways  of  approximating   these parameters using  simple statistics, such as degrees, in our direct method or a subset of random samples in our online method. 

\textbf{Random Sampling over Joins} The problem of random sampling over a single join path was posed in the 1990s~\cite{joinsynopses}. Acharya et al. proposed a solution for good approximate answers  using only random samples from the base relations, but accuracy still remained to be improved~\cite{joinsynopses}. Joining random samples of joins  produces a much smaller sample size than samples. Moreover, 
it is shown that join samples obtained do not satisfy the  independence requirement~\cite{HuangYPM19}. 
To solve this, Olken proposed the idea of rejecting join of two  samples with specific probabilities for two-table join~\cite{Olken}; Chaudhuri et al. proposed techniques that are applicable  to linear joins but not to arbitrary joins~\cite{Chaudhuri:1999}. Both methods  require full information of the tables as well as the index structure. Chaudhuri et al. significantly improved the efficiency by proposing another strategy group sample algorithm that relies on only partial statistics~\cite{Chaudhuri:1999}. However, all the above three methods only work for 2-table Joins. 
Ripple join returns dependent and uniform samples~\cite{ripple}.  
Wander join~\cite{2016_wander_join} extended  ripple join to return  independent but non-uniform samples from the join. 
Recently, Zhao et al.  proposed a framework that handles general multi-way joins and guarantees i.i.d~\cite{2018_sample_join_revisit}. This  algorithm can be plugged in our framework for random sampling over a single join path. 

\eat{
\noindent\textbf{Join Size Estimation} results on join size upper
bounds [6, 18, 22]
There has always been an interest in the database theory community in upper bounding a join size, due to the connection of this problem to query optimization. We described Chaudhuri et al.'s  algorithm~\cite{Chaudhuri:1999} and Olken’s algorithm~\cite{olken1995random}, which we extended for multi-relation joins, 
as well as Zhao et al.'s~\cite{2018_sample_join_revisit} algorithm in \S~\ref{sec:joinsize}. 
Another well-known algorithm is the AGM bound~\cite{agm}. On a general query, it requires solving a linear program. AGM returns  the optimal bound on
the join size when the only information is the size of relations.} 

\noindent\textbf{Union of Sets and Queries}  The union-of-sets problem has been studied in approximate counting literature~\cite{KarpLM89}. The goal is to design a randomized algorithm that can output an approximation of the size of the union of sets efficiently.  
Karp et al. proposed a $(1+\epsilon)$-randomized approximation algorithm
for approximating the size of the union of sets with a  linear running time. This algorithm requires the exact size of each set and a uniform random sample of each set.~\cite{KarpLM89}. Bringmann and Friedrich later applied this algorithm in designing an algorithm for
high dimensional geometric objects using uniform random sampling. They also proved that the problem is \#P-hard for high dimensional boxes~\cite{BringmannF08}. The computation of union of sets also has links to $0$-th frequency moment estimation~\cite{AlonMS96}. 
One line of work in this area is on DNF counting problem~\cite{KarpL83}, including designing hashing-based algorithms~\cite{ChakrabortyMV16, DagumKLR00, KarpLM89, MeelSV17}. 
Another popular line of work is on estimating the union of sets where  each set arrives in a streaming fashion~\cite{Bar-YossefJKST02,KaneNW10,MeelV021,Meel0V22}.

\section{Conclusion}
\label{sec:conclusion}

This paper studies two  novel problems: sampling over the union of joins and size approximation of the union of joins. A general union sampling framework is proposed that estimates join  overlap and union parameters when (1) data  statistics are available in DBMSs and (2)  access to the data in relations is feasible. The framework extends to the union size of joins of  arbitrary multi-way acyclic and cyclic. 
Interesting future work directions include analyzing the impact of data skew on approximations as well as integrating a union sampling operator into a database engine.

\clearpage

\bibliographystyle{ACM-Reference-Format}
\bibliography{ref}

\end{document}